%% file: main.tex
\documentclass{ws-ijqi}

\usepackage{blindtext}

\usepackage[english]{babel}

\input{package}
\input{quantum_header}

\input{newHeader}

\newcounter{myctr}
\def\myitem{\refstepcounter{myctr}\bibfont\noindent\ifnum\themyctr>9\else\phantom{0}\fi\hangindent17pt\themyctr.\enskip}

\begin{document}

\title{QUANTUM SOFT-COVERING LEMMA WITH APPLICATIONS TO RATE-DISTORTION CODING, RESOLVABILITY AND IDENTIFICATION VIA QUANTUM CHANNELS
}


\author{TOUHEED ANWAR ATIF {\&} S. SANDEEP PRADHAN}

\address{Department of Electrical Engineering and Computer Science,\\
University of Michigan, Ann Arbor, MI 48109, USA.\\
\texttt{\{touheed,pradhanv\}@umich.edu}}

\author{ANDREAS WINTER}

\address{ICREA {\&} Grup d'Informaci\'o Qu\`antica, Departament de F\'isica,\\ Universitat Aut\`onoma de Barcelona, 08193 Bellaterra (Barcelona), Spain.\\[2mm]
Institute for Advanced Study, Technische Universit\"at M\"unchen,\\ Lichtenbergstra{\ss}e 2a, 85748 Garching, Germany.\\[2mm]
QUIRCK--Quantum Information Independent Research Centre Kessenich,\\ Gerhard-Samuel-Stra{\ss}e 14, 53129 Bonn, Germany.\\
\texttt{andreas.winter@uab.cat}}

\maketitle

\begin{history}
\received{24 July 2023}
\revised{20 April 2024}
\end{history}

\begin{abstract}
We propose a quantum soft-covering problem for a given general quantum channel and one of its output states, which consists in finding the minimum rank of an input state needed to approximate the given channel output. We then prove a one-shot quantum covering lemma in terms of smooth min-entropies by leveraging decoupling techniques from quantum Shannon theory. 
This covering result is shown to be equivalent to a coding theorem for rate distortion under a posterior (reverse) channel distortion criterion by two of the present authors. 
Both one-shot results directly yield corollaries about the i.i.d.~asymptotics, in terms of the coherent information of the channel.  

The power of our quantum covering lemma is demonstrated by two additional applications: first, we formulate a quantum channel resolvability problem, and provide one-shot as well as asymptotic upper and lower bounds. 
Secondly, we provide new upper bounds on the unrestricted and simultaneous identification capacities of quantum channels, in particular separating for the first time the simultaneous identification capacity from the unrestricted one, proving a long-standing conjecture of the last author.
\end{abstract}

\markboth{Touheed Anwar Atif, S. Sandeep Pradhan and Andreas Winter}{Quantum Soft Covering Lemma with Applications\ldots}



\section{Introduction}
\label{QSC:sec:intro}
\input{intro_covering}

\section{Preliminaries and notations}
\label{QSC:sec:prelim}
\input{prelim}

\input{new_main_results}



\input{applications_results}


\section{Proofs of quantum source coding results}
\label{sec:lossy_proofs}
\subsection{Proof of Proposition \ref{prop:petz_reference}}
\label{app:petz_reference}
\input{petz_reference}

\input{lemma_proofs}

\subsection{Proof of the converse in Theorem \ref{thm:CoveringOneShotLossy}}
\label{sec:oneshot_converse_lossy}
\input{lossyConverseOneShot}



\section{Conclusions and future work} \label{sec:conclusion}
This work formulates the problem of quantum soft-covering in its most natural setting, and presents a one-shot characterization of the problem in terms of smoothed one-shot quantum entropic quantities. By leveraging the one-shot result, we provide a single-letter characterization of the optimal rate of the quantum soft-covering problem in terms of the minimal coherent information. We also contribute to the study of second-order asymptotics and provide achievability bounds for the same. As a part of future work, we aim to study optimal second order performance limits for the formulated quantum soft-covering problem. 
Furthermore, then we explore three applications of the quantum soft-covering lemma, including the characterization of a one-shot lossy quantum source coding problem, the formulation of a quantum channel resolvability problem, and separations between simultaneous and unrestricted capacities for identification via quantum channels. 
For the channel resolvability, we provide a lower bound, and conjecture that it is equal to the  quantum transmission capacity of the given quantum channel.

\section*{Acknowledgements}
The authors thank David Neuhoff and Mohammad Aamir Sohail for their contributions and discussions on the subject matter of the present paper, and Fr\'ed\'eric Dupuis for discussions and the permission to include his proof of Lemma \ref{lemma:H-max-min} in the present paper. AW furthermore acknowledges earlier conversations with Paul Cuff and Zahra Khanian regarding a different approach to quantum soft covering. 

TAA and SSP are supported in part by National Science Foundation grant CCF-2007878.
AW is supported by the European Commission QuantERA grant ExTRaQT 
(Spanish MICINN project PCI2022-132965), by the Spanish MICINN 
(project PID2019-107609GB-I00) with the support of FEDER funds, 
by the Spanish MICINN with funding from European Union NextGenerationEU 
(PRTR-C17.I1) and the Generalitat de Catalunya, by the Alexander 
von Humboldt Foundation, and the Institute for Advanced Study 
of the Technical University Munich.

\appendix
\section{Proof of Lemma \ref{lemma:H-max-min} (courtesy of Fr\'ed\'eric Dupuis)}
\label{sec:proof-Dupuis}
We start with the observation that Lemma \ref{lemma:H-max-min} is a consequence of the following more general statement. 

\begin{mylemma}
Let $\rho^{AB}$ be a state on $A\otimes B$, and $\sigma^B$ a state on $B$ with full support. Then, 
\[
  H_{\max}^{\sqrt{1-\epsilon^4}}(A|B)_{\rho|\sigma}
   \leq H_{\min}^\epsilon(A|B)_{\rho|\sigma}.
\]
\end{mylemma}
\noindent
Here, recall the definitions of the conditional min- and max-entropies relative to a fixed state \cite{tomamichel2012framework}:
\[\begin{split}
  H_{\min}^\epsilon(A|B)_{\rho|\sigma} 
   &= -\log \min t \text{ s.t. } \widetilde{\rho} \leq t\1_A\otimes\sigma^B,\ P(\rho,\widetilde{\rho}) \leq \epsilon,\ \Tr\widetilde{\rho}\leq 1,\\
  H_{\max}^\delta(A|B)_{\rho|\sigma} 
   &= \phantom{-} 
     \log \min F(\widetilde{\rho},\1_A\otimes\sigma^B) \text{ s.t. } P(\rho,\widetilde{\rho}) \leq \delta, \ \Tr\widetilde{\rho}\leq 1.
\end{split}\]
The familiar min- and max-entropies we used in the article result from further optimisation of the state $\sigma^B$. 

\begin{proof}
The above definition of $H_{\min}^\epsilon(A|B)_{\rho|\sigma}$ is a semidefinite program (SDP) \cite{Boyd-SDP}, as we can see when expanding it using a purification $\psi^{ABC}$ of $\rho^{AB}$: 
\[\begin{split}
  2^{-H_{\min}^\epsilon(A|B)_{\rho|\sigma}}
   &= \min t \text{ s.t. } \widetilde{\rho}^{AB} \leq t\1_A\otimes\sigma^B,\ \Tr\widetilde{\rho}^{ABC}\psi^{ABC} \geq 1-\epsilon^2,\\
   &\phantom{\min t \text{ s.t. }}
   \widetilde{\rho}^{ABC}\geq 0,\ \Tr\widetilde{\rho}\leq 1. 
\end{split}\]
Its dual SDP reads 
\[\begin{split}
  \max\, (1-\epsilon^2)\lambda-\mu \text{ s.t. }
   \Tr\sigma^B E_B \leq 1,\ 
   \lambda\psi^{ABC} \leq E_{AB}\otimes\1_C + \mu\1_{ABC},\ E_{AB}\geq 0.
\end{split}\]
Because strong duality holds for this pair of mutually dual SDPs, they have the same values. From now on fix primal and dual optimal solutions, $t$, $\widetilde{\rho}^{ABC}$ and $\lambda$, $\mu$, $E_{AB}$, respectively. Furthermore, let $P_{AB}$ be the projector onto the support of $E_{AB}$ and denote $P^\perp = \1_{AB}-P$. 

To start, from $\lambda\psi^{ABC} \leq E_{AB}\otimes\1_C + \mu\1_{ABC}$ we get, by sandwiching left and right hand side with $P^\perp\otimes\1_C$: 
\[
  \lambda(P^\perp\otimes\1_C)\psi^{ABC}(P^\perp\otimes\1_C)
   \leq \mu(P^\perp\otimes\1_C),
\]
and hence, observing that the left hand side here has rank one, $\lambda\Tr\rho^{AB}P^\perp \leq \mu$. Now, since by optimality of the solutions $t = (1-\epsilon^2)\lambda-\mu \geq 0$, it follows that 
$\lambda\Tr\rho^{AB}P^\perp \leq \mu \leq (1-\epsilon^2)\lambda$. That is,
\begin{equation}
  \label{eq:rho-P-trace}
  \Tr\rho^{AB}P^\perp \leq 1-\epsilon^2,
  \quad
  \Tr\rho^{AB}P \geq \epsilon^2. 
\end{equation}

Consider now the subnormalised state $\rho'=P\rho P$ on $AB$ with purification $\psi'=(P\otimes\1_C)\psi(P\otimes\1_C)$. We will show that its max-entropy relative to $\sigma^B$ is bounded by $H_{\min}^\epsilon(A|B)_{\rho|\sigma}$ and that it is within $\sqrt{1-\epsilon^4}$ of $\rho$ in purified distance. Beginning with the former, we have
\[\begin{split}
  2^{H_{\max}(A|B)_{\rho'|\sigma}} 
   &=  F\left(P\rho P, \1_A\otimes\sigma^B\right) \\
   &\leq F\left(P\rho P, P(\1_A\otimes\sigma^B)P+P^\perp(\1_A\otimes\sigma^B)P^\perp\right) \\
   &= F\left(P\rho P, P(\1_A\otimes\sigma^B)P\right) \\
   &= \frac{1}{t} F\left(P\rho P, P\widetilde{\rho}P\right) 
    \leq \frac{1}{t} 
    = 2^{H_{\min}^\epsilon(A|B)_{\rho|\sigma}}.
\end{split}\]
Here, the second line follows from the monotonicity of the fidelity under the CPTP pinching map of $P$ and $P^\perp$; the fourth line follows from complementary slackness \cite{Boyd-SDP}, $\widetilde{\rho}^{AB} E_{AB} = t(\1_A\otimes\sigma^B)E_{AB}$, which carries over from $E$ to its support projector $P$. 

On the other hand, invoking Eq.~\eqref{eq:rho-P-trace},
\[\begin{split}
  F(\rho,\rho') 
   &\geq F(\psi,\psi')
    =    \Tr\psi(P\otimes\1_C)\psi(P\otimes\1_C) \\
   &=    \left|\bra{\psi}(P\otimes\1_C)\ket{\psi}\right|^2 \\
   &=    \left(\Tr\rho^{AB}P\right)^2
    \geq \epsilon^4,
\end{split}\]
and hence $P(\rho,\rho') \leq \sqrt{1-\epsilon^4}$ as claimed.
\end{proof}






\bibliographystyle{unsrturl}
\bibliography{references}

\end{document}

%% file: package.tex
\usepackage{bbm}
\usepackage{physics}
\usepackage{amssymb}
\usepackage{amsmath}
\usepackage{graphicx}
\usepackage{hyperref}
\hypersetup{
    colorlinks=true,
    allcolors=red,
    linkcolor=red,
    filecolor=black,      
    urlcolor=black, 
    }
\usepackage[super,sort,compress]{cite}
\usepackage{setspace}
\usepackage{revsymb}
\usepackage{mathrsfs}

%% file: quantum_header.tex
\newtheorem{mydefinition}{Definition}
\newtheorem{mylemma}[mydefinition]{Lemma}
\newtheorem{mycorollary}[mydefinition]{Corollary}
\newtheorem{mytheorem}[mydefinition]{Theorem}
\newtheorem{myremark}[mydefinition]{Remark}
\newtheorem{myproposition}[mydefinition]{Proposition}

\newtheorem{myconjecture}[mydefinition]{Conjecture}

\global\long\def\RR{\mathbb{R}}
\global\long\def\CC{\mathbb{C}}

\global\long\def\EE{\mathbb{E}}

\def \calA {\mathcal{A}}
\def \calB {\mathcal{B}}

\def \calD {\mathcal{D}}
\def \calE {\mathcal{E}}

\def \calH {\mathcal{H}}

\def \calL {\mathcal{L}}
\def \calM {\mathcal{M}}
\def \calN {\mathcal{N}}

\def \calP {\mathcal{P}}

\def \calS {\mathcal{S}}

\def \RR {\mathbb{R}}

\newcommand{\sq}[1]{\left[{#1}\right]}
\newcommand{\round}[1]{\left({#1}\right)}

\def \bV {W}
\def \bA {{A_R}}
\def \bB {{B_R}}

\def \bAn {A_R^n}
\def \bBn {B_R^n}

\def\tensor{\otimes}

\def \Wchannel {\calN}

\usepackage{stackengine}
\def\deq{\mathrel{\ensurestackMath{\stackon[1pt]{=}{\scriptstyle\Delta}}}}

\def\hmin{H_{\min}}
\def\hmax{H_{\max}}

%% file: newHeader.tex
\def\squareforqed{\hbox{\rlap{$\sqcap$}$\sqcup$}}
\def\qed{\ifmmode\squareforqed\else{\unskip\nobreak\hfil
\penalty50\hskip1em\null\nobreak\hfil\squareforqed
\parfillskip=0pt\finalhyphendemerits=0\endgraf}\fi}
\def\endenv{\ifmmode\;\else{\unskip\nobreak\hfil
\penalty50\hskip1em\null\nobreak\hfil\;
\parfillskip=0pt\finalhyphendemerits=0\endgraf}\fi}

\mathchardef\ordinarycolon\mathcode`\:
\mathcode`\:=\string"8000
\def\vcentcolon{\mathrel{\mathop\ordinarycolon}}
\begingroup \catcode`\:=\active
  \lowercase{\endgroup
  \let :\vcentcolon
  }

\newcommand{\nc}{\newcommand}
\nc{\proj}[1]{| #1\rangle\!\langle #1 |}
\nc{\avg}[1]{\langle#1\rangle}
\nc{\smfrac}[2]{\mbox{$\frac{#1}{#2}$}}
\nc\ox{\otimes}
\nc{\dg}{\dagger}
\nc{\dn}{\downarrow}
\nc{\cA}{{\cal A}}
\nc{\cB}{{\cal B}}
\nc{\cC}{{\cal C}}
\nc{\cD}{{\cal D}}
\nc{\cE}{{\cal E}}
\nc{\cF}{{\cal F}}
\nc{\cG}{{\cal G}}
\nc{\cH}{{\cal H}}
\nc{\cI}{{\cal I}}
\nc{\cJ}{{\cal J}}
\nc{\cK}{{\cal K}}
\nc{\cL}{{\cal L}}
\nc{\cM}{{\cal M}}
\nc{\cN}{{\cal N}}
\nc{\cO}{{\cal O}}
\nc{\cP}{{\cal P}}
\nc{\cQ}{{\cal Q}}
\nc{\cR}{{\cal R}}
\nc{\cS}{{\cal S}}
\nc{\cT}{{\cal T}}
\nc{\cU}{{\cal U}}
\nc{\cW}{{\cal W}}
\nc{\cX}{{\cal X}}
\nc{\cY}{{\cal Y}}
\nc{\cZ}{{\cal Z}}
\nc{\csupp}{{\operatorname{csupp}}}
\nc{\qsupp}{{\operatorname{qsupp}}}
\nc{\Var}{{\operatorname{Var}}}
\nc{\rar}{\rightarrow}
\nc{\lrar}{\longrightarrow}
\nc{\polylog}{{\operatorname{polylog}}}
\nc{\1}{{\openone}}

\nc\id{{\operatorname{id}}}



%% file: intro_covering.tex
Covering and packing are two fundamental principles in information theory that are essential to the design and analysis of coding systems achieving the fundamental limits. Naturally, these concepts are also important in the framework of quantum information theory and have been developed in analogy with the classical case.
Indeed, packing problems have been extensively studied in the quantum information literature, mostly related to the design of classical, quantum or hybrid error-correcting codes \cite{holevo1998capacity,schumacher1997sending,gottesman2002introduction,lidar2013quantum}. Problems of a covering nature have shown up naturally in other contexts, most prominently in wiretap and similar cryptographic settings \cite{CaiWinterYeung:q-wiretap,devetak2005private,DevetakWinter:key-PRL}, but also in channel and correlation synthesis and, crucially, converse theorems to identification coding \cite{ahlswede2002strong,Sen-oneshot-covering,winter,bennett2014quantum,ChenWinter:EoP,shen2022optimal}.

The origin of this line of research can be traced back to the work of Wyner in not one but two articles \cite{Wyner:wiretap,Wyner:CR}, the first introducing the idea of randomizing over a channel code to confound a wiretapper on a communication channel, the second defining and solving the task of distributed agents generating classical i.i.d. samples of local random variables $X$ and $Y$ distributed according to a joint distribution $P_{XY}$ from shared randomness; the optimal rate of initial shared randomness is now called Wyner's common information. 
In the classical setup, this has been further developed toward the problem of distributed channel synthesis in the classical reverse Shannon theorem \cite{crst}, and the subsequent unifying work of Cuff \emph{et al.} \cite{cuff2010coordination,cuff2013distributed}.

In the quantum setting, Ahlswede and Winter formulated a classical-quantum soft-covering lemma, and showed that classical randomness at a rate of the Holevo information (with respect to a classical-quantum ensemble) can generate the output of a classical quantum channel \cite{ahlswede2002strong}.
Evolving from earlier works \cite{massar2000amount,winter2001compression}, Winter \cite{winter} successfully performed an information theoretic study of quantum measurements and formulated
the fundamental task of compressing or ``faithfully simulating'' quantum measurements. 
An instrumental tool in obtaining the result was a the classical-quantum soft-covering lemma. 
In both classical and  quantum settings, the problem of soft-covering has also found  applications in the quantum reverse Shannon theorem \cite{bennett2014quantum}. Recent contributions to the study of classical-quantum soft-covering also include Refs.~\citen{cheng2022error,shen2022optimal,atif2021distributed,anshu2019convex}. In particular Refs.~\citen{atif2021distributed,anshu2019convex} addressed soft-covering lemmas for pairwise independent ensembles. 

In the present article, we formulate and address the problem of quantum soft-covering in its natural generality. The task 
can be described in simple terms as follows: given a quantum channel $\calN$ and one of its output states $\calN(\rho)$, what is the smallest rank of a quantum state $\sigma$ that reproduces or approximates the given output state, $\calN(\sigma) \approx \calN(\rho)$? The approximation here is quantified using the trace distance.

To this end, we first consider the one-shot version of the problem, where we characterize the minimum rank in terms of smoothed quantum entropic quantities \cite{tomamichel2015quantum}. Based on the one-shot result, we then study the quantum soft-covering problem in an asymptotic setting, where the objective is to approximate the output produced by $n$ independent uses of the given channel, acting on a memoryless source. We obtain a single-letter characterization of the asymptotically optimal rate of the quantum soft-covering problem in terms of the coherent information, and then also study the finite block length behavior to second order. More precisely, we obtain an achievability bound for the specific portion of the performance limit that scales proportionally as $1/\sqrt{n}$. 
Such a study was first introduced by Strassen in 1962 for the problem of classical communication \cite{strassen1962asymptotische}, and has been the topic of interest for many subsequent works \cite{hayashi2009information,polyanskiy2010channel,yassaee2013technique,tomamichel2013hierarchy,li2014second,beigi2014quantum,beigi2016decoding}.

In order to showcase the power of our quantum soft-covering lemma, we turn our attention to three fundamental problems in quantum information theory.
The first is the lossy quantum source coding problem under a posterior (reverse) channel distortion criterion formulated in Ref.~\citen{atif2023lossy}, where a single-letter characterization of the performance limit in terms of the coherent information was given.   
Using the quantum covering results, we provide a one-shot characterization of this problem.
Moreover, using the one-shot results, we study the 
asymptotic regime of the lossy source coding problem and recover the  asymptotic performance limit characterized in Ref.~\citen{atif2023lossy}. 


As the next application, we consider the quantum analogue of the channel resolvability problem. 
We formulate a quantum channel resolvability problem and 
show that the quantum channel resolvability rate for an arbitrary channel is upper bounded by its strong converse quantum transmission capacity. Using the one-shot converse to the soft-covering problem, we provide a lower bound to the quantum resolvability rate in terms of smooth minimum entropy. This bound also forms a new lower bound  in the classical setting, where a complete characterization of the problem is still an open question \cite{hayashi2006general}. We conjecture that, in the asymptotic setting, this lower bound is equal to the quantum transmission capacity of the channel.
As the final application, we consider the task of classical identification via quantum channels, and we are able to answer some of the open questions relating to this task, utilizing the quantum soft-covering tool. 
In detail, we obtain an upper bound on the so-called \emph{simultaneous} identification capacity of a general quantum channel $\calN$, which also forms a strong converse upper bound. This bound resolves an open question concerning the separation of the simultaneous and the unrestricted identification capacities. 
Finally, for a general quantum channel, we provide another upper bound on the unrestricted identification capacity in terms of the strong converse quantum transmission capacity of the channel. 

\medskip
The rest of the paper is organized as follows. We provide some necessary definitions and useful lemmas in Section \ref{QSC:sec:prelim}.
In Section \ref{QSC:sec:mainResults_covering}, we formulate the quantum soft-covering problem and provide the main results pertaining to the one-shot (Theorem \ref{thm:CoveringOneShotSmooth}), asymptotic (Theorem \ref{thm:asympCovering}) and second order characterizations (Theorem \ref{thm:second_order}). In Sections \ref{QSC:sec:mainResults_lossy}, \ref{QSC:sec:mainResults_resol}, and \ref{QSC:sec:mainResults_identification}, we move on to the three information theoretic applications of the covering result: the lossy source coding problem (Theorems \ref{thm:CoveringOneShotLossy} and \ref{thm:mainResult_lossy}), the channel resolvability problem (Theorem \ref{thm:resolvability}) and identification via quantum channels (Theorems \ref{thm:C-ID-sim-of-qubit} and \ref{thm:C-id-upperbound}), respectively. The proofs for the source coding theorems are provided in the Section \ref{sec:lossy_proofs}.
In Section \ref{sec:conclusion} we conclude.


%% file: prelim.tex
We supplement the notations from Ref.~\citen{wilde_arxivBook} with the following. Let $\1_A$ denote the identity operator acting on a Hilbert space $A$. The set of density operators on $A$ is denoted $\calD(A)$, linear operators by $\calL(A)$, and positive semi-definite operators by $\calP(A)$. The set of sub-normalized states is denoted by  $\calD_{\leq}(A) \deq \{\rho \in \calP(A): \Tr \rho \leq 1\}.$ 
We denote by $\bA$ the Hilbert space associated with the reference space of $A$, isomorphic to $A$: $\bA \simeq A$, in particular $\dim \bA = \dim A$. 
The dimension of a Hilbert space $A$ is denoted $|A|$ (recalling the cardinality of a set, and coinciding with the cardinality of any basis of $A$). 
For a linear operator $\omega$, let $|\omega| = |\operatorname{supp}\omega|$ denote its rank, namely the dimension of its support (usually reserved for Hermitian and often even semidefinite operators).
The fidelity between two states, $\rho,\sigma \in \calD(A)$, is defined as \[
F(\rho,\sigma) \deq \|\sqrt{\rho}\sqrt{\sigma}\|_1^2 = \max_{\phi, \psi} |\langle \phi | \psi \rangle|^2,
\]
where $\phi$ and $ \psi$ range over purifications of $\rho$ and $\sigma$, respectively. 
The fidelity between two sub-normalized states $\rho,\sigma \in \calD_\leq(A)$ is defined as 
\[F_{\star}(\rho,\sigma) \deq \left[ \Tr{|\sqrt{\rho}\sqrt{\sigma}|} + \sqrt{(1-\Tr\rho)(1-\Tr\sigma)}\right]^2.\]
We define 
the purified distance as $P(\rho,\sigma) \deq \sqrt{1-F_\star(\rho,\sigma)}$.  
Following the notation of Ref.~\citen{tomamichel2015quantum}, for a given sub-normalized state $\rho \in \calD_{\leq}(A)$, we use $\mathscr{B}^\epsilon(A;\rho)$ to denote the $\epsilon$-ball of sub-normalized states around $\rho$ defined as 
\[
\mathscr{B}^\epsilon(A;\rho) \deq \{\sigma \in \calD_\leq(A): P(\sigma,\rho) \leq \epsilon\},
\]
and when clear from the context, will omit the system and use $\mathscr{B}(\rho).$
For a density operator $\rho \in \calD(A)$, the von Neumann entropy is defined as $S(\rho) \deq -\Tr \rho\log\rho$. For a bipartite state $\rho^{AB} \in \calD(A \tensor B)$, the conditional entropy is defined as $S(A|B)_\rho \deq S(\rho^{AB}) - S(\rho^B)$, where $\rho^B \deq \Tr_A{\rho^{AB}}$, and the quantum mutual information $I(A:B)_\rho \deq S(\rho^A)+S(\rho^B)-S(\rho^{AB})$. 
For a given CPTP map $\calN:A \rightarrow B$, and a state $\rho \in \calD(B)$, we define the set $ \calS(\rho,\calN) \deq \{\sigma \in \calD(A): \calN(\sigma) = \rho\}$ and its smoothed variant  $\calS_\epsilon(\rho,\calN) \deq \{\sigma \in \calD(A): \|\rho - \calN(\sigma)\|_1 \leq \epsilon\}$.
For any real $x$, we define $x^+ \deq \max(x,0)$, which for a more complex argument 
 is written as $[f(x)]^+$. 

\subsection{Useful definitions}

We recall the following definitions of the well-known information quantities.



\begin{mydefinition}[Coherent information]
For a CPTP map $\calN: A \rightarrow B$, and an input density operator $\rho \in \calD(A)$, the \emph{coherent information}
of $\calN$ with respect to $\rho$ is defined as   
\[
I_c(\rho,\calN) \deq I_c(A_R \rangle B)_\omega \deq 
 - S(A_R|B)_{\omega},
\]  
where $\omega^{B A_R} \deq (\calN \otimes \id) \phi^{A A_R}_\rho$, and $\phi_{\rho}^{AA_R}$ is an arbitrary purification of $\rho_A$. 
\end{mydefinition}

\begin{mydefinition}[Holevo information]
    Given an ensemble $\{P(x),W^B_x\}_{x \in \mathcal{X}}$, where $\mathcal{X}$ is  a finite set and $W^B_x \in \mathcal{D}(B)$ are states, its \emph{Holevo information} is defined as
    \[
      I(X:B) \deq S\left(\sum_{x\in \mathcal{X}}P(x) W^B_x\right) - \sum_{x\in \mathcal{X}}P(x)S(W^B_x).
    \]
    It is the quantum mutual information of the classical-quantum state $\rho^{XB}\deq\sum_x P(x)\ketbra{x}^X \otimes W_x^B$.
\end{mydefinition}

\begin{mydefinition}[Quantum information variance]
For any two density operators $\rho,\sigma \in \calD(A)$, such that $\mbox{supp} \rho \subset \operatorname{supp} \sigma$, the \emph{quantum information variance} is defined as 
\[
V( \rho \| \sigma) \deq \Tr[\rho(\log \rho- \log \sigma)^2]- D^2(\rho \|\sigma).
\]
\end{mydefinition}

\begin{mydefinition}[Min- and max-entropy \cite{tomamichel2015quantum}]
    For $\rho^{AB} \in \calD_{\leq}(A\tensor B)$, the \emph{min-} and \emph{max-entropy} of $A$ conditioned on $B$ of the (sub-normalized) state $\rho^{AB}$ is defined as 
    \begin{align*}
        \hmin(A|B)_\rho &\deq \sup_{\sigma^B \in \calD_{\leq}(B)} \sup\{\lambda \in \RR: \rho^{AB} \leq 2^{-\lambda} \1_A \tensor \sigma^B \},  \\
        \hmax(A|B)_\rho &\deq \max_{\sigma^B \in \calD_{\leq}(B)} \log F(\rho^{AB},\1_A \tensor \sigma^B),
    \end{align*}
    respectively.
\end{mydefinition}

\begin{mydefinition}[Smoothed entropies \cite{tomamichel2015quantum}]
    For $\rho^{AB} \in \calD_{\leq}(A\tensor B)$, and $\epsilon \in (0,1)$,  the \emph{$\epsilon$-smooth min-} and \emph{max-entropy} of $A$ conditioned on $B$ of the (sub-normalized) state $\rho^{AB}$ is defined as 
    \begin{align*}
        \hmin^\epsilon(A|B)_\rho \deq \max_{\sigma\in \mathscr{B}^\epsilon(\rho)}\hmin(A|B)_{\sigma} \quad \text{and} \quad
        \hmax^\epsilon(A|B)_\rho \deq \min_{\sigma \in \mathscr{B}^\epsilon(\rho)}\hmax(A|B)_{\sigma},
    \end{align*}
    respectively.
\end{mydefinition}

\begin{mydefinition}[Smoothed max-relative entropy \cite{tomamichel2013hierarchy}]
    For a given $\rho\in \calD(A)$ and $\omega \in \calP(A)$, and $\epsilon \in (0,1)$, we define the \emph{smooth max-relative entropy} as
    \begin{align*}
        D_{\max}^\epsilon(\rho \| \omega) \deq \min_{\sigma \in \mathscr{B}^\epsilon(\rho)} \inf \{\lambda: \sigma \leq 2^\lambda\omega  \}.
    \end{align*}
\end{mydefinition}

\subsection{Useful lemmas}
\label{sec:usefulLemmas}
We need the following results that relate fidelity and trace norm.
\begin{mylemma}[{Fuchs/van~de~Graaf\cite{fuchs1999cryptographic}, see also~\cite[Thm.~9.3.1]{wilde_arxivBook}}]
\label{lem:relationshipTraceFidelity}
    For any two states $\rho,\sigma \in \calD(A)$, we have
    \begin{equation*}
      1-\sqrt{F(\rho,\sigma)} \leq \frac{1}{2} \norm{\rho-\sigma}_1 \leq \sqrt{1-F(\rho,\sigma)}.
    \end{equation*}
\end{mylemma}

\begin{mylemma}[{Canonical~purification~\cite[Lemma~2]{atif2023lossy}, \cite[Lemma~14]{winter}}] 
\label{lem:closenessofPurification}
For $\rho, \sigma \in \calD(A)$, the following inequality holds:
\begin{equation*}
  F(\psi_{\rho}, \psi_{\sigma}) \geq \left(1-\frac{1}{2}\norm{{\rho}-{\sigma}}_1\right)^2,
\label{eq:fidelityPurification1}
\end{equation*}
where $\ket{\psi_{\rho}}$ and $\ket{\psi_{\sigma}}$ are the canonical purifications of $\rho$ and $\sigma$, respectively:
\[
  \ket{\psi_\rho} = \left(\sqrt{\rho}\otimes\1\right)\ket{\Gamma},
\]
with the unnormalized maximally entangled vector
$\ket{\Gamma} = \sum_{i} \ket{i}\ket{i}$.
\end{mylemma}

Next we state the asymptotic equipartition property (AEP) for the max-relative entropy. 

\begin{mylemma}[{AEP \cite[Eq.~(35)]{tomamichel2013hierarchy}}] \label{lem:hypoRelativeSecondOrder}
    For any $n \geq 1$, and $\epsilon \in (0,1)$, we have  
    \begin{align*}
      \label{eq:Dh_nD}
    D_{\max}^{\epsilon}(\rho^{\tensor n} \| \sigma^{\tensor n}) =n D(\rho \| \sigma)-\sqrt{n V(\rho \| \sigma)} \Phi^{-1}(\epsilon^2) +O(\log n),    
    \end{align*}
    where $\Phi(\delta)$ denotes the cumulative distribution function of the standard normal distribution.
\end{mylemma}

We need a couple of tools from the 
min-entropy calculus \cite{tomamichel2015quantum}.
\begin{mylemma}[{Min-~vs.~max-entropy~\cite[Lemma 5]{vitanov2013chain}, \cite[Prop.~5.5]{tomamichel2012framework}}]
  \label{lemma:H-min-max}
  For any bipartite state $\rho^{AB}\in \calD(A\tensor B)$ and $\epsilon+\delta < 1$,
  \begin{equation*}
    H_{\min}^\epsilon(A|B)_\rho \leq H_{\max}^\delta(A|B)_\rho +\log\frac{1}{1-(\epsilon+\delta)^2}.
  \end{equation*}
\end{mylemma}

\begin{mylemma}[{Dupuis 2013, cf.~\cite[Lemma 9]{Q-prettystrong}}]
  \label{lemma:H-max-min}
  For any bipartite state $\rho^{AB}\in \calD(A\tensor B)$ and $0<\epsilon<1$,
  \begin{equation*}
    H_{\min}^\epsilon(A|B)_\rho \geq H_{\max}^\delta(A|B)_\rho,
  \end{equation*}
  where $\delta = \sqrt{1-\epsilon^4}$.
\end{mylemma}
\begin{proof}
We reproduce Fr\'ed\'eric Dupuis' original proof in \ref{sec:proof-Dupuis}.
\end{proof}

%% file: new_main_results.tex
\section{Quantum soft covering}
\label{QSC:sec:mainResults_covering}
We begin by formulating the 
quantum covering problem as follows. 

\begin{mydefinition}[{Quantum covering}] \label{def:qc covering setup}
A quantum covering setup is characterized by a pair $(\rho^B,\Wchannel)$, where $\rho^B \in \mathcal{D}(B)$ is a density operator 
and $\Wchannel$ is a CPTP map from $A$ to $B$.
\end{mydefinition}
A code for quantum covering is defined as follows.
\begin{mydefinition}[Quantum covering code]
\label{def:qu covering code}
For a given pair $(\rho^B,\Wchannel)$,
an $(n,\Theta,\epsilon)$-code for quantum covering is a density operator 
$\hat\sigma^{A^n}$ on $A^{\otimes n}$ such that 
$|\hat\sigma^{A^n}|=\Theta$, and 
\begin{align*}
    \frac12 \left\| \rho^{B ^{\otimes n}} - \Wchannel^{\otimes n}(\hat\sigma^{A^n}) \right\|_1 \leq \epsilon.
\end{align*}
\end{mydefinition}


\noindent As our first main result, we characterize the smallest rank $\Theta$ for a one-shot $(1,\Theta,\epsilon)$ quantum covering code in terms of smoothed min-entropy. 

\begin{mytheorem}[One-shot quantum covering]\label{thm:CoveringOneShotSmooth}
Given a pair $(\rho^B,\Wchannel)$, 
and a density operator $\sigma^{A} \in \mathcal{S}(\rho^B,\Wchannel)$, for all $\delta,\eta \in (0,1)$, there exists a $(1,\Theta,\epsilon)$-code for quantum covering such that
\begin{align}
    \label{eq:soft-covering-direct}
        \log\Theta \leq  \sq{-\hmin^{\delta}(A_R|B)_{\omega} - 2\log{\eta}}^+,
\end{align}
and $\epsilon \leq 8(\delta+\eta)$,
where $\omega^{B A_R}=(\Wchannel \otimes \id)\Phi^{A A_R}_\sigma$,  
$\Phi^{A A_R}_\sigma$ is a purification of $\sigma^{A}$.
Moreover, for all $\epsilon \in (0,1)$, every $(1,\Theta, \epsilon)$-code satisfies 
\begin{equation}
    \label{eq:soft-covering-converse}
    \log\Theta \geq \inf_{\sigma^A \in \mathcal{S}_\epsilon(\rho^B,\Wchannel)} \left[ -H_{\min}(A_R|B)_{\omega} \right]^+.
\end{equation}
\end{mytheorem}

\begin{proof}
    \input{oneShotSmooth}
\end{proof}

Considering the asymptotic setting, we define achievable rate for a quantum covering code as follows.
\begin{mydefinition}[Achievability]
\label{def:achievability_covering}
For a given pair $(\rho^B,\Wchannel)$, a rate $R$ is said to be achievable for quantum covering, if for all $ \epsilon > 0 $ and all sufficiently large $n$, there exists an $ (n, \Theta,\epsilon)$ quantum code such that 
$ \frac{1}{n}\log{\Theta} \leq R + \epsilon$.
\end{mydefinition}

\noindent Building on Theorem \ref{thm:CoveringOneShotSmooth}, we now characterize the set of all achievable rates using single-letter coherent information.


\begin{mytheorem}[Asymptotic quantum covering]\label{thm:asympCovering}
     For a given pair $(\rho^B,\Wchannel)$, a rate $R$ is achievable for quantum covering if and only if $S(\rho^B,\Wchannel)$ is non-empty, and
    \[R \geq \min_{\sigma^{A} \in \calS(\rho^B, \Wchannel)} \left[I_c(\sigma^{A},\Wchannel)\right]^+.
    \]
\end{mytheorem}
\begin{proof}

\input{asymptoticConverse.tex}
\end{proof}

\begin{myremark}
The matrix tail bounds \cite{ahlswede2002strong} and their subsequent developments \cite{tropp:book,Sen-oneshot-covering} can be thought of as soft-covering results for classical-quantum (cq) channels, modelled as maps $W:\cX\rightarrow \cD(B)$ acting as $\cX \ni x \mapsto W_x \in \cD(B)$. Given the channel and a probability distribution $P$ on the (discrete) input alphabet $\cX$, we are looking for the minimum $N$ such that $W(P) = \sum_x P(x) W_x$ is approximated by an empirical average $\frac1N \sum_{i=1}^N W_{x_i}$. In the i.i.d.~limit of channels $W^{\ox n}$, the optimal rate $\frac1n\log N$ approaches $I(X:B)$, the Holevo information of the ensemble $\{P(x),W_x\}$ \cite{ahlswede2002strong}; in the one-shot setting, the role of the mutual information is taken by a suitable one-shot information \cite{Sen-oneshot-covering}. 

The reason that there we see a mutual information and not a coherent information appearing is that the problem statement restricts us to a sample from the given discrete alphabet, as the input of the cq-channel is literally a classical system. To compare to the above definitions and results, we would need to interpret $W$ as the CPTP map $\cN(\rho) = \sum_x \bra{x}\rho\ket{x} W_x$ acting on the input quantum system $X=\CC^{|\cX|}$, and identify $P$ with the input state $\rho=\sum_x P(x)\proj{x}$. However, our Definitions \ref{def:qc covering setup} and \ref{def:qu covering code}, as well as Theorems \ref{thm:CoveringOneShotSmooth} and \ref{thm:asympCovering} do not restrict the approximating input $\sigma$ to density matrices diagonal in the classical basis $\{\ket{x}\}$, and that accounts for the different rate behaviour, now governed by the coherent information. 
\end{myremark}

As our next main result, we obtain a second order achievability bound. 
The asymptotic result in Theorem \ref{thm:asympCovering} states that the rank of the smallest quantum covering code scales exponentially with the exponent $nI_c(\cdot)$. The objective of the second order analysis is to obtain a more precise estimate of the exponent
by deriving the coefficient of $\sqrt{n}$ present in the exponent. The exploration of this field was initially introduced by Strassen in 1962, specifically addressing the classical communication problem.
This was later refined by other authors  \cite{hayashi2009information,polyanskiy2010channel}.
In the quantum domain, Tomamichel and Hayashi \cite{tomamichel2013hierarchy}, and Li \cite{li2014second}  
were the first to obtain tight results concerning second order asymptotics.
See also Refs.~\citen{beigi2014quantum,beigi2016decoding}, where new proofs toward studying the second order behaviour, inspired from Ref.~\citen{yassaee2013technique}, 
are given utilizing the collision relative entropy. 

\begin{mytheorem}[Second order achievability bound] 
\label{thm:second_order}
For a given $(\rho^B,\Wchannel)$ quantum covering setup, 
and a density operator $\sigma^{A} \in \mathcal{S}(\rho^B,\Wchannel)$, 
there exists an $(n,\Theta,\epsilon)$ covering code such that the rate scales as
\begin{align*}
    \frac{1}{n}\log{\Theta} = \sq{I_c(\sigma^{A},\calN) - \frac{1}{\sqrt{n}}\sqrt{V(\sigma^{A},\calN)}\Phi^{-1}\bigg(\frac{\epsilon^2}{100}\bigg) + O\left( \frac{\log n}{n} \right)}^+,   
\end{align*}
where 
\[
    V(\sigma^{A},\calN) \deq V(\rho^{B A_R} \| \rho^B \otimes \1_{A_R}),
\]
$\rho^{BA_R} \deq (\calN\tensor \id )\Phi_\sigma^{A \bA}$, and $\Phi_\sigma^{A\bA}$ is a purification of $\sigma^A$. 
\end{mytheorem}
\begin{proof}
\input{oneShotCollisionCovering}
\end{proof}

%% file: oneShotSmooth.tex
We begin with the showing the direct part, Eq.~\eqref{eq:soft-covering-direct}.
We apply the general decoupling theorem \cite[Thm.~3.1]{dupuis2014one} 
to $\omega^{B A_R}$, a family of random unitaries $U$ on $A$ forming at 
least a $2$-design, and the fixed map $\calM:A_R\rightarrow X$ defined by 
\[
  \calM(\alpha) = \sum_{x=1}^\ell (\Tr\alpha P_x)\proj{x},
\]
where the $\ket{x}$ form an orthonormal basis of $X$ and 
 the $P_x$ are mutually orthogonal subspace projectors of rank $r$ 
each forming a POVM, $\sum_x P_x = \1_{A_R}$. For this to exist, we need 
$r\ell=|A_R|$, which we can, without loss of generality, assume by isometrically enlarging the 
Hilbert space. The Choi matrix $\mu^{A_RX}$ of this channel has 
$H_2(A_R|X)_\mu=\log r$ \cite[Corollary 5.10, see Eq.~(5.67)]{tomamichel2015quantum}, and so \cite[Thm.~3.1]{dupuis2014one} (see 
also Ref.~\citen{berta2009single}) implies 
\[
  \EE_U \frac{1}{2}\left\| (\cM\circ\cU\ox\id_B)\omega^{A_RB} - \frac{\1_X}{|X|}\ox\omega^B \right\|_1 
        \leq 2\epsilon + 2^{-\frac12\log r -\frac12 H_{\min}^\delta(A_R|B)_\omega}, 
\]
where $\cU(\alpha) = U\alpha U^\dagger$ is the channel applying the unitary $U$, and only the second term is smoothed, resulting in a lower factor. 
If we fix $\log r = -H_{\min}^\delta(A_R|B)_\omega - 2\log\eta$, the right hand side 
will be bounded from above by $\eta$, and allowing for rounding to a smaller integer, 
we obtain the bound $2\eta$. 

In particular, in that case there exists a unitary $U$ such that 
\[
  \frac{1}{2}\left\| (\cM\circ\cU\ox\id_B)\omega^{A_RB} - \frac{\1_X}{|X|}\ox\omega^B \right\|_1 
        \leq 2\delta + 2\eta,
\]
and to progress we expand the two matrices inside the trace norm. Namely, 
inserting $\omega^{A_RB}=(\id_{A_R}\ox\calN)\Phi_\sigma^{A_RA}$ and in particular 
$\omega^B=\calN(\sigma^A)$, we get
\[\begin{split}
  \frac{\1_X}{|X|}\ox\omega^B
     &= \sum_{x=1}^\ell \frac{1}{|X|} \proj{x}^X \ox \calN(\sigma^A) \text{ and } \\
  (\cM\circ\cU\ox\id_B)\omega^{A_RB}
     &= \sum_{x=1}^\ell \proj{x}^X \ox \Tr_{A_R}\left[\omega^{A_RB}\left(U^\dagger P_x U\ox\1_B\right)\right] \\
     &= \sum_{x=1}^\ell \proj{x}^X \ox 
                        \calN\left(\Tr_{A_R}\left[\Phi_\sigma^{A_RA}\left(U^\dagger P_x U\ox\1_A\right)\right]\right) \\
     &= \sum_{x=1}^\ell p(x) \proj{x}^X \ox \calN\left(\sigma_x^A\right), 
\end{split}\]
where $p(x)\sigma_x^A = \Tr_{A_R}\left[\Phi_\sigma^{A_R A}\left(U^\dagger P_x U\ox\1_A\right)\right]$, 
$p(x)$ is the normalising trace (a probability) and $\sigma_x^A$ is a state 
of rank at most $r$ on $A$.

By tracing over $B$ we find that the total variational distance between 
$p$ and the uniform distribution is less than $ 2\epsilon+2\eta$, and hence by 
the triangle inequality we find 
\[
  \sum_{x=1}^\ell p(x) \frac{1}{2}\left\| \calN(\sigma_x^A) - \calN(\sigma^A) \right\|_1 \leq 4\delta + 4\eta.
\]
To conclude, there exists at least one $x$ such that 
$ \frac{1}{2}\left\| \calN(\sigma_x^A) - \calN(\sigma^A) \right\|_1 \leq 4\delta + 4\eta$,
making the statement of the theorem true letting $\hat\sigma^A = \sigma_x^A$.



As for the converse, Eq.~\eqref{eq:soft-covering-converse}, consider a purification $\Phi^{AA_R}_\sigma$ of $\sigma^A$. Then we have
\[\begin{split}
  \log\rank\rho^A &\geq H_{\max}(A_R)_\Phi \\
                  &=    -H_{\min}(A_R|A)_\Phi \\
                  &\geq -H_{\min}(A_R|B)_{\omega}, 
\end{split}\]
where the first line follows from the fact that $H_{\max}(\sigma^A)$
is the R\'enyi-$\frac12$ entropy, which is upper bounded by the 
R\'enyi-$0$ entropy; the second line is by the duality of min- and 
max-entropies \cite[Prop.~5.7]{tomamichel2015quantum}, and the third comes from data processing. 
Since the rank of a density matrix is a natural number, we also 
have $\log\rank\sigma^A \geq 0$.  

%% file: asymptoticConverse.tex
Invoking the asymptotic equipartition property for the min-entropy \cite{tomamichel2015quantum},
we immediately get the following consequence which provides the achievability.
  Given a pair $(\rho^B,\calN)$ and an input state $\sigma^A \in \calS(\rho^B,\calN)$, 
  we consider the i.i.d.~extension $\calN^{\ox n}:A^n \rightarrow B^n$.
  Then, using Theorem \ref{thm:CoveringOneShotSmooth}, we see that there exists a sequence of states 
  $\hat\sigma^{A^n}$ on $A^n$ of rank $r_n$ such that 
  \begin{equation*}
    \frac12 \left\| \calN^{\ox n}(\hat\sigma^{A^n}) - \round{\rho^B}^{\tensor n} \right\|_1 \rightarrow 0
    \quad\text{and}\quad
    \frac1n \log r_n \rightarrow \left[-S(A_R|B)_\omega\right]^+,
  \end{equation*}
  as $n\rightarrow\infty$, where $\omega^{A_R B} \deq (\id\tensor \calN)\Phi^{A_R A}_\sigma$, and $\Phi_\sigma^{A_R A}$ is a purification of $\sigma^A$.
  \qed
For the converse, we proceed as follows. Let $R$ be achievable. Using the result from Theorem \ref{thm:CoveringOneShotSmooth}, 
for any $\epsilon>0$, we get
\begin{align*}
    nR &\geq \inf_{\sigma^{A^n} \in \calS_\epsilon((\rho^{B})^{\tensor n},\calN^{\tensor n})}-\hmin(A^n_R|B^n)_{\upsilon^n(\sigma^{A^n})},
\end{align*}
where $\upsilon^n(\sigma^{A^n}) \deq ( \id \tensor\calN^{\tensor n} )\Phi_{\sigma}^{A^n_R A^n}$ with $\Phi_{\sigma}^{A^n_R A^n}$ denoting a purification of $\sigma^{A^n}$.
Using the fact that $\hmin(A^n_R|B^n)_{\upsilon^n(\sigma^{A^n})} \leq H(A^n_R|B^n)_{\upsilon^n(\sigma^{A^n})}$ \cite[Lemma 2]{tomamichel2009fully},  we get 
\begin{align}
    nR &\geq \inf_{\sigma^{A^n} \in \calD(A^{\tensor n}):\|(\rho^{B})^{\tensor n} - \calN^{\tensor n}(\sigma^{A^n})\|_1 \leq \epsilon}-H(A^n_R|B^n)_{\upsilon^n(\sigma^{A^n})} 
    \label{eq:covering_converse} \\
    &=   \inf_{\sigma^{A^n} \in \calD(A^{\tensor n}):\|(\rho^{B})^{\tensor n} - \calN^{\tensor n}(\sigma^{A^n})\|_1 \leq \epsilon} H(B^n)_{\upsilon^n(\sigma^{A^n})} -H(A^n_R,B^n)_{\upsilon^n(\sigma^{A^n})} \nonumber \\
    &\stackrel{a}{\geq}   \inf_{\sigma^{A^n} \in \calD(A^{\tensor n}):\|(\rho^{B})^{\tensor n} - \calN^{\tensor n}(\sigma^{A^n})\|_1 \leq \epsilon} H(B^n)_{(\rho^{B})^{\tensor n}} -H(A^n_R,B^n)_{\upsilon^n(\sigma^{A^n})} - n\tilde{\epsilon}_1 \nonumber \\
    &\stackrel{b}{\geq} n \min_{\sigma^{A} \in \calD(A): \| \rho^{B}-\mathcal{N}(\sigma^{A}) \|_1 \leq \epsilon} \nonumber
I_c(\sigma^{{A}},\calN)  -2n\tilde{\epsilon}_1, 
\end{align}
where 
$(a)$ follows from the Fannes-Audenart inequality, the constraint on $\sigma^{A^n}$ and defining $\tilde{\epsilon}_1\deq \epsilon\log|B| + h_b(\epsilon)$, and   $(b)$ follows from employing the same reasoning as that used in \cite[Eqs.~(68)-(72)]{atif2023lossy}. The converse proof is completed by using the continuity arguments provided in Ref.~\citen{atif2023lossy}.

%% file: oneShotCollisionCovering.tex
Given the pair $(\rho^B,\calN)$ and $\sigma^{A} \in \calS(\rho^B,\calN)$, consider the $n$-letter pair $((\rho^B)^{\tensor n},\calN^{\tensor n})$, and let $\sigma^{A^n}\deq (\sigma^A)^{\tensor n}$. From Theorem \ref{thm:CoveringOneShotSmooth}, it follows  that for all $\delta,\eta \in (0,1)$ there exists a $(n,\Theta,8(\delta+\eta))$  quantum covering code such that $\log \Theta  \leq  \sq{-\hmin^{\delta}(\bAn|B^n)_{\rho^{B^n\bAn}} - 2\log\eta}^+$, where $\rho^{B^n\bAn} \deq (\calN^{\tensor n} \tensor \id)\Phi_{\sigma}^{A^n \bAn}$, and $\Phi_{\sigma}^{A^n \bAn}$ is a purification of $\sigma^{A^n}$.
We now bound $\log \Theta$ as
\begin{align*} 
  \log \Theta & \leq \sq{-\hmin^{\delta}(\bAn|B^n)_{(\rho^{B\bA})^{\tensor n}} - 2\log\eta}^+, \nonumber\\
 & \overset{a}{=} \sq{\inf_{\tau^{B^n} \in \calD(\calH_B^{\tensor n})}D_{\max}^{\delta}({(\rho^{B\bA})^{\tensor n}}\| \1_{\bAn}\tensor \tau^{B^n}) - 2\log\eta}^+ \nonumber \\
& \leq \sq{D_{\max}^{\delta}({(\rho^{B\bA})^{\tensor n}}\| \1_{\bAn}\tensor (\rho^B)^{\tensor n}) - 2\log\eta}^+
\end{align*}
where $(a)$ follows from the definition of $\hmin^{\delta}$ \cite[Def.~3]{tomamichel2013hierarchy}. 
Choosing $\eta= \delta/4$ and 
using Lemma \ref{lem:hypoRelativeSecondOrder} gives the desired result.

%% file: applications_results.tex
We now move on to considering the three applications of the quantum covering problem. 

\section{Lossy quantum source coding}
\label{QSC:sec:mainResults_lossy}

The history of this problem dates back to the work of Barnum \cite{barnum2000quantum}, who introduced the idea of a local distortion criterion using entanglement fidelity and formulated a quantum rate distortion problem. 
Barnum conjectured the minimal coherent information of the quantization (forward) channel to characterize the asymptotic performance limit of this problem. However, in Ref.~\citen{datta2012quantum} a regularised expression using the concept of entanglement of purification was obtained. Further notable investigations for related scenarios involving additional resources have been performed in studies subsequently \cite{wilde2013auxiliary,datta2013quantum,datta2013one,khanian2021rate,khanian2022general,baghali2022rate,koashi2001compressibility,devetak2002quantum,winter2002compression,hsieh2016channel,salek2018quantum,anshu2019convex}. Recent progress was made in Ref.~\citen{atif2023lossy} where the use of a posterior (reverse) channel distortion criterion to quantify distortion was proposed, together with a global error criterion instead of the conventional average distortion measure. This formulation, which was inspired by certain approximation results in classical source and channel coding problems \cite{shamai1997empirical, pradhan2004approximation}, provided a more optimistic viewpoint on the rate-distortion problem, as the formulation yielded the single-letter minimal coherent information of the posterior reference channel as the asymptotic performance limit. 
In this work, we aim to characterize the performance limit of a one-shot formulation for the aforementioned rate distortion problem. To begin, we establish the following notation.

For a given Hilbert space $A$, its reference $A_R$, and two orthonormal basis $\{\ket{i}\}^A$ and $\{\ket{i}\}^{A_R}$ corresponding to the Hilbert spaces $A$ and $A_R$, respectively,  we define the unnormalized maximally entangled state $\Gamma_{A_R A}$ as $\sum_i \ket{i}^{A_R} \tensor \ket{i}^{A}$. 
The transpose of a state $\tau \in \calL(A)$  
is defined as $\tau^T \deq \sum_{i,i'} \ket{i}\bra{i'}^{A_R} \bra{i'} \tau \ket{i}^A$.
Here, we focus exclusively on references obtained from canonical purifications of quantum states \cite[Lemma 14]{winter}, and for a given density operator $\sigma^A \in \calD(A)$, define the canonical purification $\ket{\psi_\sigma}^{\bA A}$ of $\sigma^A$ as $\ket{\psi_\sigma}^{\bA A} \deq (\1_{\bA}\tensor \sqrt{\sigma^A})\Gamma_{\bA A}$. We use $\Psi_\sigma^{\bA A}$ to denote the density operator corresponding to $\ket{\psi_\sigma}^{\bA A}$. Note that $\sigma^{A_R} \deq \Tr_{A}\round{{\Psi_\sigma}^{\bA A}} = \round{\sigma^{A}}^T$.
As {is the convention}, for two states acting on the same Hilbert space, we use the same $\Gamma$ when defining their canonical purifications. 
Recalling the notion of a posterior reference map, we have the following.

\begin{figure}
    \centering
    \includegraphics[trim={0 0 0 0},clip,width=\textwidth]{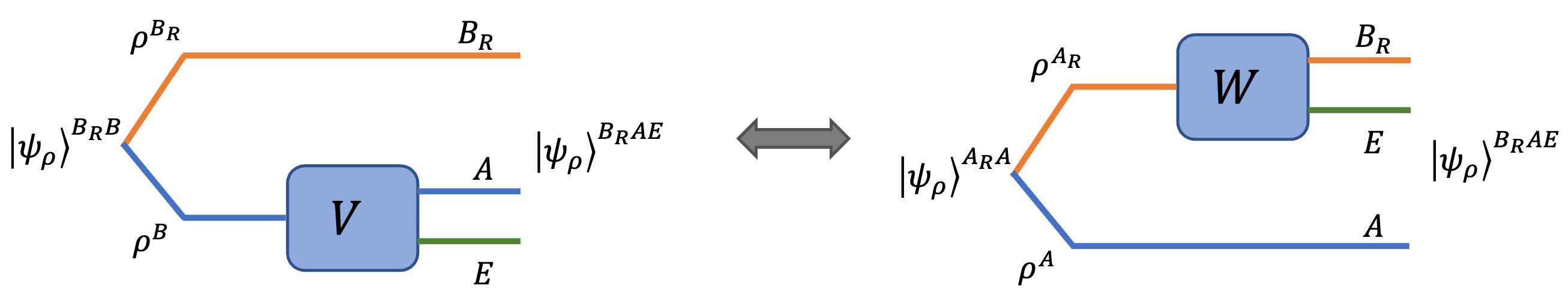}
    \caption{Figure demonstrating the construction of the posterior reference map $W$ from the isometry $V$ (the Stinespring's dilation of $\calN_V$) and the source state $\rho^B$.}
    \label{fig:reverseChannel}
\end{figure}

\begin{mylemma}[Posterior reference map \cite{atif2023lossy}]\label{def:VBar}
    Given a source $\rho^B \in \calD(B)$ and a channel $\calN_V: {B} \rightarrow A$, let $\rho^A \deq \calN_V(\rho^B)$.  Let $V: {B} \rightarrow {A}\otimes E$ be a Stinespring's isometry corresponding to the CPTP map $\calN_V$ with $\dim(E) \geq \dim(A)$, such that $\calN_V(\cdot) = \Tr_{E}\{V(\cdot)V^\dagger\}$. 
    Then there exists a unique CPTP map $\calN_{W}:{\bA} \rightarrow {\bB}$, with an associated isometry $W:{\bA} \rightarrow {\bB} \otimes E$, which we refer to as the \emph{posterior reference map} of $V$ with respect to  $\rho^A$, satisfying 
    \begin{align}
    \label{eq:WV_postrefmap}
        (W \otimes \1_{A})\ket{\psi_\rho}^{\bA A} = (\1_{\bB}\tensor V)\ket{\psi_\rho}^{\bB B},
    \end{align} 
    where $\ket{\psi_\rho}^{\bA A}$ and $\ket{\psi_\rho}^{\bB B}$ are the canonical purifications of $\rho^A$ and $\rho^B$, respectively, as shown in Figure \ref{fig:reverseChannel}.  
\end{mylemma}

  The posterior reference map exhibits a interesting connection with the well-known Petz recovery map \cite{petz1986sufficient}. In particular, the posterior reference map with respect to a state is equal to a transposed Petz recovery channel acting on the reference Hilbert space of the given state as stated below. 
  
\begin{myproposition}[Posterior reference map as a Petz recovery map] \label{prop:petz_reference}
    Given an input state $\rho^B$ and a channel $\calN_V$, for an arbitrary state $\sigma^{A_R}$, we have
    \[
    \calN_W(\sigma^{A_R}) = (\rho^{\bB})^{1/2} \round{\calN_V^\dagger\round{(\rho^A)^{-1/2} \sigma^A(\rho^A)^{-1/2}}}^T(\rho^{\bB})^{1/2}, 
    \]
    where $\rho^A = \calN_V(\rho^B)$, $\sigma^A = \Tr_{\bA}{\Psi_{\sigma}^{A \bA}}=\round{\sigma^{A_R}}^T$, $\rho^{B_R} = \Tr_{B}{\Psi_{\rho}^{B \bB}}$.
\end{myproposition}
\begin{proof}
    A proof of this statement is provided in Section \ref{app:petz_reference}.
\end{proof}


 Continuing with the source coding setup, we proceed by recalling the formulation of the problem \cite{atif2023lossy}.

\begin{mydefinition}[{Quantum source coding setup}] \label{def:qc source coding setup}
A quantum source coding setup is characterized by a pair $(\rho^B,\calN_W)$, where $\rho^B \in \mathcal{D}(B)$ is a density operator,  $A$ is a reconstruction Hilbert space, and $\calN_W$ is a CPTP map from ${\bA}$ to ${\bB}$, where ${\bA}$ and ${\bB}$ are reference spaces corresponding to $A$ and $B$, respectively.
\end{mydefinition}

\begin{mydefinition}[Lossy quantum compression protocol] \label{def:protocolcompression}
For a given input and reconstruction Hilbert spaces $(B,A)$, 
an $(n,\Theta,\epsilon)$ lossy quantum compression protocol (as shown in Figure \ref{fig:q_protocol}) consists of an encoding CPTP map $\calE^{(n)}:{B^{\tensor n}} \rightarrow {M}$  and a decoding CPTP map $\calD^{(n)}:M \rightarrow {A^{\tensor n}}$, such that $\dim M = \Theta$, and
\begin{align} 
\label{def:protocolError}
\frac12 \left\|{\omega^{\bBn A^n} - (  \calN_W^{\otimes n} \tensor \id_{A^n} ) \Psi_{\omega}^{\bAn A^n }}\right\|_1 \leq \epsilon,   
\end{align}
where $\omega^{\bBn A^n} \deq (\id \otimes \calD^{(n)})(\id\otimes \calE^{(n)}) (\Psi_{\rho}^{\bBn B^n })$, and $\Psi_{\rho}^{\bBn B^n}$ and $\Psi_{\omega}^{\bAn A^n } $ are the canonical purifications of ${\rho^B}^{\otimes n}$ and $\omega^{A^n}$, respectively.
\end{mydefinition}

\begin{figure}
    \centering
    \includegraphics[scale=0.7]{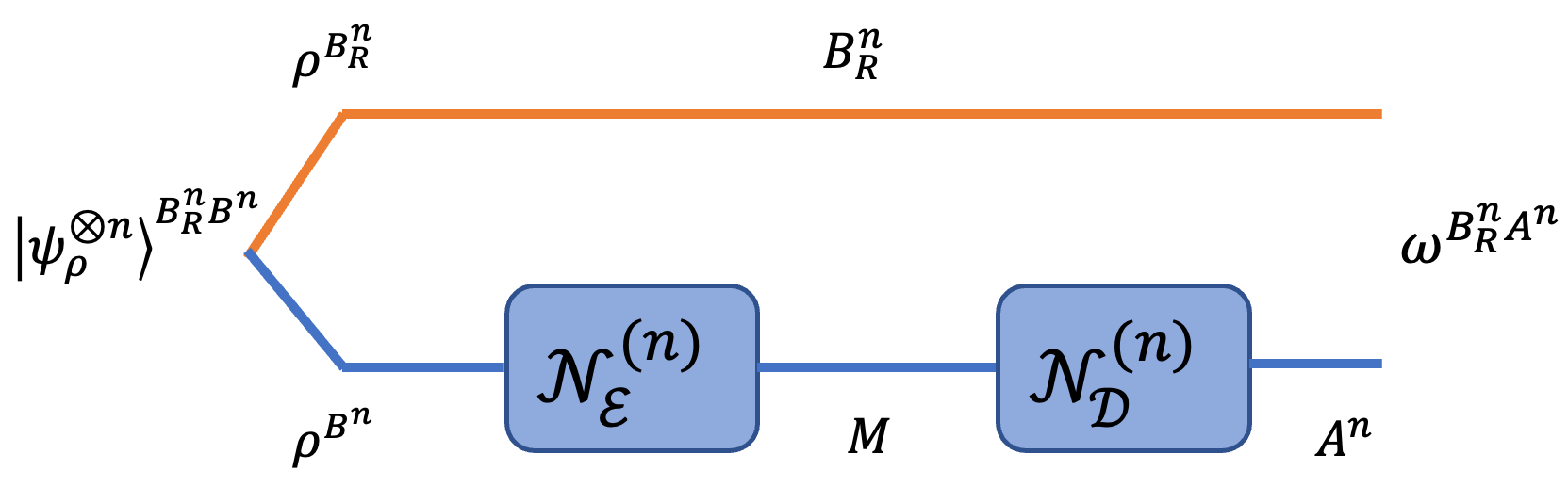}
    \caption{Illustration of Lossy Quantum Compression protocol}
    \label{fig:q_protocol}
\end{figure}

\begin{mydefinition}[Achievability]
\label{def:achievability}
For a quantum source coding setup $(\rho^B,\calN_W)$, a rate $R$ is said to be achievable for lossy quantum source coding, if for all $ \epsilon > 0 $ and all sufficiently large $n$, there exists an $ (n, \Theta,\epsilon) $ lossy quantum compression protocol satisfying
$ \frac{1}{n}\log{\Theta} \leq R + \epsilon$.
\end{mydefinition}

As a first result of this section, we aim to provide  a characterization to the one-shot lossy source coding problem. 
To achieve this, we rely on the properties of posterior reference maps \cite{atif2023lossy}, which enable us to establish the following equivalence between the quantum covering problem and the lossy quantum source coding problem.



\begin{mylemma}[Equivalence of quantum lossy source coding and quantum covering] \label{lem:coveringAndcompression}
Consider an arbitrary state  $\rho^B \in \calD(B)$ and a CPTP map $\calN_W:A_R \rightarrow B_R$. 
The existence of an $(n,\Theta,\epsilon)$ lossy quantum compression protocol 
for the  
$(\rho^B,\calN_W)$ setup,
with the restriction that the decoding map ${\calD}^{(n)}$ is an isometry, 
implies that
there exists  an $(n,\Theta,\epsilon)$ quantum covering code for the $(\rho^{\bB}, \calN_W)$ setup, where $\rho^{\bB}=\round{\rho^B}^T$.
Conversely,  the existence of 
an $(n,\Theta,\epsilon)$  quantum covering  code for $(\rho^{B_R},\calN_W)$ setup implies that 
there exists  an $(n,\Theta,4\sqrt{\epsilon})$ lossy quantum compression protocol for the $(\rho^B, \calN_W)$ setup.
\end{mylemma}
\begin{proof}
A detailed proof is provided in Section \ref{proof:coveringAndcompression}. The proof involves utilizing a quantum soft-covering code to construct an encoder for the source coding problem. This encoder enables efficient compression of a quantum source while maintaining the quantified level of loss.
As for the other direction, we find that an encoder for the source coding problem possesses the ability to generate a state that satisfies the soft-covering constraint. 
\end{proof}

In Ref.~\citen{atif2023lossy}, the duality between quantum source coding and channel coding problems was highlighted using the structure of the performance limit: both being characterized using the coherent information. In the following, we take a step further, and shed light  on the duality relation between the encoder of the former and the decoder of latter, and vice versa, which is well understood in the classical settings \cite{csiszar2011information,pradhan2003duality}.
Expanding on this line of inquiry, we note that Ref.~\citen{beigi2016decoding} presents a decoder based on the Petz recovery map for the quantum channel coding problem.
As can be noted from the proof of the above equivalence (Lemma \ref{lem:coveringAndcompression}), an encoder for the current source coding problem can be devised using the posterior reference map associated with the quantum soft covering setup. Further, Proposition \ref{prop:petz_reference} uncovers the connection between the posterior reference map and the Petz recovery map, thus establishing 
 a compelling duality between the encoding and decoding maps of the two problems.


Leveraging the one-shot results pertaining to quantum covering and the above equivalence, we are able to characterize the achievability results for the performance of a one-shot lossy source coding problem formulation in the following. The converse, however, uses a different chain of arguments.

\begin{mytheorem}[One-shot lossy quantum compression]
    \label{thm:CoveringOneShotLossy}
For a $(\rho^B,\calN_W)$ quantum source coding setup, a density operator $\sigma^{\bA} \in \mathcal{S}(\rho^{\bB},\calN_W)$, for all $\delta,\eta \in (0,1)$, there exists $(1,\Theta,\epsilon)$ lossy quantum compression protocol such that
    \[
    \log\Theta \leq  \sq{-\hmin^{\delta}(A|B_R)_{\rho^{B_R A}} - 2\log{\eta}}^+,
    \]
    and $\epsilon \leq 4\delta+4\eta$,
    where 
 $\rho^{B_R A}=(\calN_W \otimes \id)\Phi^{A_R A}_\sigma$,  
$\Phi^{A A_R}_\sigma$ is a purification of $\sigma^{A_R}$.
Moreover, for all $\epsilon \in (0,1)$, every $(1,\Theta, \epsilon)$ code satisfies 
\begin{equation*}
    \log\Theta \geq \inf_{\sigma^\bA \in \mathcal{S}_\epsilon(\rho^\bB,\calN_W)} \left[ -H^{\sqrt{\epsilon}}_{\min}(A|B_R)_{\rho^{B_R A}} \right]^+.
\end{equation*}
\end{mytheorem}

\begin{proof}
    The proof of achievability follows from using the equivalence with soft-covering, according to Lemma \ref{lem:coveringAndcompression} and the achievability part of Theorem \ref{thm:CoveringOneShotSmooth}. A proof of the converse is provided in Section \ref{sec:oneshot_converse_lossy}.
\end{proof}

This gives us the following characterization to the set of all achievable rates for the asymptotic lossy quantum source coding problem, recovering the characterization obtained in \cite[Thm.~1]{atif2023lossy}.

\begin{mytheorem}[Asymptotic lossy quantum compression]
    \label{thm:mainResult_lossy}
    For a $(\rho^B,\calN_W)$ quantum source coding setup, a rate $R$ is achievable  if and only if $S(\rho^\bB,\calN_W)$ is non empty, and
    \[R \geq \min_{\sigma^{\bA} \in \calS(\rho^\bB, \calN_W)}I^+_c(\sigma^{\bA},\calN_W).
    \]
\end{mytheorem}
\begin{proof}
    The proof of achievability follows by using the above one-shot achievability result of Theorem \ref{thm:CoveringOneShotLossy} and the asymptotic equipartition property (similar to the proof of Theorem \ref{thm:asympCovering}). The proof of the converse is provided in Section \ref{sec:converse_asymp_lossy}.
\end{proof}

\section{Quantum channel resolvability}
\label{QSC:sec:mainResults_resol}

As a next application, we consider the channel resolvability problem.
Initially introduced in the classical setting by Han and Verd\'u \cite{HanVerdu}, and generalised to classical-quantum channels with a given output measurement by L\"ober \cite{Loeber:PhD}, and in Ref.~\citen{ahlswede2002strong} for the plain output density matrices, this problem addresses the task of asymptotically approximating a target output distribution through a given channel using an essentially uniform input distribution with a smaller support. The capacity of this problem is defined as the smallest rate of input randomness with which any arbitrary output distribution can be effectively approximated.
Analogously, for the quantum setting, we define its capacity as the rate of the smallest rank of input state needed to approximate any arbitrary output state of a given channel.
The objective now is to characterize this smallest rank. To formally define the problem, we first define channel resolvability for a given input state $\sigma^{A^n}  \in \calD(A^{\otimes n})$. 




\begin{mydefinition}[Channel resolvability for a given input state]
For a given CPTP map $\calN:A \rightarrow B$,
a rate $R$ is said to be $(n,\epsilon)$-achievable for
an input state $\sigma^{A^n} \in \calD(A^{\otimes n})$ 
if there exists a density operator 
$\hat\sigma^{A^n}$ on $A^{\otimes n}$ with
$|\hat\sigma^{A^n}|=2^{nR}$, which satisfies
\begin{align*}
    \frac12 \left\| \cN^{\tensor n}(\hat\sigma^{A^n}) - \cN^{\tensor n}(\sigma^{A^n}) \right\|_1 \leq \epsilon.
\end{align*}
We define
\[
\mathcal{R}^0_{\epsilon}(\cN^{\tensor n},\sigma^{A^n}) \deq \inf \{R: R \mbox{ is } (n,\epsilon)\mbox{-achievable for } \sigma^{A^n}\}.
\]
The asymptotic resolvability rate is defined accordingly as
\[
  \mathcal{R}(\cN) \deq \sup_{\epsilon>0} \limsup_{n\rightarrow\infty} \left[ \sup_{\sigma^{A^n}} \mathcal{R}^0_{\epsilon}\left(\cN^{\ox n},\sigma^{A^n}\right)
\right].
\]
\end{mydefinition}

\begin{mydefinition}[Quantum transmission code]
    For a given CPTP map $\calN: A \rightarrow B$, a $(n,\Theta,\epsilon)$ quantum transmission code consists of a pair of encoding and decoding maps, $\calE:  C\rightarrow A^{\tensor n}$ and $\calD:B^{\tensor n} \rightarrow C$, respectively, such that 
   $|C|=\Theta$, and  \[
    \sup_{\ket\Phi \in C\tensor {C_R}} P\left(\ketbra\Phi, (\id_{C_R}\tensor \calD\circ \calN^{\tensor n}\circ \calE)\ketbra\Phi\right) \leq \epsilon.
    \]
\end{mydefinition}


\begin{mydefinition}[Weak converse quantum capacity]
    For a given channel $\calN$, we define $N_E(n,\epsilon|\cN)$ as the maximum dimension $\Theta$ such that there exists a $(n,\Theta,\epsilon)$ quantum transmission code. The quantum capacity is now defined as 
    \[
    Q(\calN) \deq \inf_{\epsilon > 0} \liminf_{n\rightarrow \infty} \frac{1}{n} \log N_E(n,\epsilon|\calN).
    \]
\end{mydefinition}

\begin{mydefinition}[Strong converse quantum capacity]
    For a given channel $\calN$, we define
    $R$ to be admissible if all sequences of
    $(n,\Theta_n,\epsilon_n)$ quantum transmission codes
    with 
    \[
    \liminf_{n \rightarrow \infty} \frac{1}{n} \log \Theta_n >R, 
    \]
    satisfy
    \[
    \lim_{n \rightarrow \infty} \epsilon_n =1.
    \]
    The strong converse capacity is defined as
\[
\widehat{Q}(\calN) \deq \inf\{R: R \mbox{ is  admissible}\}.
\]
\end{mydefinition}

The above direct part and converse for quantum soft covering yield the following
bounds on the asymptotic resolvability rates of a channel.

\begin{mytheorem}
  \label{thm:resolvability}
  For any channel $\cN$, any $\epsilon>0$, and any input state $\sigma^{A^n} \in \calD(A^{\tensor n})$, we have 
  \[
   \inf_{\omega^{A^n} \in \calD(A^{\tensor n}) \text{ s.t. }\atop \|\calN^{\tensor n}( \sigma^{A^n}) - \calN^{\tensor n}(\omega^{A^n})\|_1 \leq \epsilon} \sq{-\hmin(A_R^n|B^n)_{\rho^{B^n A_R^n}} }^+  
    \leq  
   \mathcal{R}^0_{\epsilon}\left(\cN^{\ox n},\sigma^{A^n}\right) 
    \leq \widehat{Q}(\cN),
  \]
  where $\rho^{B^n A_R^n} \deq (\calN^{\tensor n}\tensor \id_{A^n_R})(\Phi_\omega^{A^n A_R^n})$, $\Phi_\omega^{A^n A_R^n}$ is a purification of $\omega^{A^n}$, and $\widehat{Q}(\cN)$ is the 
  strong converse quantum capacity of $\cN$. 
\end{mytheorem}


\begin{proof}
\input{ProofResolvability}
\end{proof}

\begin{mycorollary}
  For any channel $\cN$, we have 
  \(
   \mathcal{R}(\cN)
    \leq \widehat{Q}(\cN),
  \)
  where $\widehat{Q}(\cN)$ is the 
  strong converse quantum capacity of $\cN$. 
\end{mycorollary}

\begin{myconjecture}
      \label{conj:resolvability}
  For any channel $\cN$,
  \(
   Q(\calN)   \leq  \mathcal{R}(\cN),
  \)
  where ${Q}(\cN)$ is the 
  weak converse quantum capacity of $\cN$. 
\end{myconjecture}



\begin{myremark}
    Note that finding an exact characterization of the asymptotic resolvability of a classical channel is still an open problem (see \cite[Chapter 6.3]{han_book}).
\end{myremark}

\begin{myremark}
The strong converse property for the quantum capacity, i.e. $Q(\cN)=\widehat{Q}(\cN)$, 
is known to hold for PPT entanglement-binding channels and for generalised 
dephasing channels \cite{TWW:q-strong}. For general degradable channels it 
is conjectured, but only a ``pretty strong'' converse is known, meaning that 
the asymptotic rate of quantum codes is bounded by $Q(\cN)$ for sufficiently 
small error \cite{Q-prettystrong}. 

For all other channels, the issue of the strong converse is wide open, 
and no Shannon-style single-letter or regularised formula is known for 
$\widehat{Q}(\cN)$. However, \cite[Thm.~8 and Cor.~7]{TWW:q-strong} establishes 
the so-called Rains information of the channel as an upper bound to $\widehat{Q}(\cN)$. 
\end{myremark}

\section{Identification via quantum channels}
\label{QSC:sec:mainResults_identification}
%
The problem of identification has a rich history, going back to the groundbreaking work of Ahlswede and Dueck \cite{AhlswedeDueck:ID}. They observed that Shannon's celebrated theory of transmission imposes a very stringent constraint on the communication task, in that the receiver has to decode the correct message among all possible ones, whose relaxation to mere ``identification'' of the correct message  leads to doubly exponentially large codes in the block length $n$. 
Specifically, given a transmitted message $m$ and an arbitrary message $m'$, the receiver's sole interest lies in determining whether ``$m = m'$'' or ``$m \neq m'$''. Han and Verd\'u \cite{HanVerdu,han1993approximation}  uncovered a fundamental connection between the channel resolvability and the classical identification  problem, where the channel resolvability rate forms an upper bound on the identification rate of the channel (cf. \cite[Lemma 6.4.1]{han_book}).


In the quantum realm, L\"ober \cite{Loeber:PhD} began the study of identification through quantum channels, uncovering a fundamental distinction inherent to the quantum setting. Unlike the classical scenario, quantum measurements face inherent incompatibility, thereby necessitating a choice regarding the receiver's capability to answer all ``if $m$ equals $m'$'' questions or only one of them. 
This leads to the introduction of simultaneous identification capacity, wherein a single measurement enables the identification of all messages simultaneously. 
Another notable generalization is the concept of quantum identification capacity of a quantum channel \cite{Winter:ID1,Winter:ID2,HaydenWinter:ID}. 
Ref.~\citen{Winter:ID-review} reviews the different notions 
of identification codes for quantum channels and the state of the art at the time of writing, which is still largely up to date.

\begin{mydefinition}[{L\"ober \cite{Loeber:PhD}}]
  \label{defi:c-ID-code}
  A {classical identification code for the channel}
  $\cN:A\rightarrow B$ with {error probability $\lambda_1$
  of first, and $\lambda_2$ of second kind} is a set
  $\{(\rho_i,D_i):i=1,\ldots,N\}$ of states $\rho_i$ on $A$ and
  operators $D_i$ on $B$ with $0\leq D_i\leq \1_B$, i.e.~the pair
  $(D_i,\1_B-D_i)$ forms a measurement, such that
  \begin{align*}
    \forall i       &\quad \Tr\bigl( \cN(\rho_i) D_i \bigr) \geq 1-\lambda_1, \\
    \forall i\neq j &\quad \Tr\bigl( \cN(\rho_i) D_j \bigr) \leq \lambda_2.
  \end{align*}
  For the special case of memoryless channels $\cN^{\otimes n}$, 
  we speak of an \emph{$(n,\lambda_1,\lambda_2)$-ID code}, and denote
  the largest size $N$ of such a code $N(n,\lambda_1,\lambda_2)$.

  An identification code as above is called \emph{simultaneous} if
  all the $D_i$ are coexistent: this means that there exists a positive
  operator valued measure (POVM) $(E_t)_{t=1}^T$ and subsets
  ${\cal D}_i\subset\{1,\ldots,T\}$ such that $D_i=\sum_{t\in{\cal D}_i} E_t$.
  The largest size of a simultaneous $(n,\lambda_1,\lambda_2)$-ID code
  is denoted $N_{\rm sim}(n,\lambda_1,\lambda_2)$.
\end{mydefinition}
  
  Note that $N_{\rm sim}(n,\lambda_1,\lambda_2) = N(n,\lambda_1,\lambda_2) = \infty$
  if $\lambda_1+\lambda_2 \geq 1$, hence to avoid this triviality one
  has to assume $\lambda_1+\lambda_2 < 1$.

\begin{mydefinition}
  \label{defi:ID-capacities}
  The \emph{(simultaneous) classical ID-capacity} of a quantum channel $\cN$
  is given by
  \begin{align*}
    C_{\rm ID}(\cN) &= \inf_{\lambda > 0} 
                       \liminf_{n\rightarrow \infty} \frac{1}{n} \log\log N(n,\lambda,\lambda), \\
    C_{\rm ID}^{\rm sim}(\cN) &= \inf_{\lambda > 0} 
                       \liminf_{n\rightarrow \infty} \frac{1}{n} \log\log N_{\rm sim}(n,\lambda,\lambda),
  \end{align*}
  respectively. 
  We say that the \emph{strong converse} holds for the identification capacity
  if for all $\lambda_1+\lambda_2 < 1$,
  \[
    \lim_{n\rightarrow \infty} \frac{1}{n} \log\log N(n,\lambda_1,\lambda_2) = C_{\rm ID}(\cN),
  \]
  and similarly for $ C_{\rm ID}^{\rm sim}$.
\end{mydefinition}

The next theorem resolves one of the main open questions around identification 
capacities of quantum channels, dating back to Ref.~\citen{ahlswede2002strong}, and explicitly conjectured in Ref.~\citen{Winter:ID1}, 
by separating the simultaneous from the unrestricted version. 

\begin{mytheorem}
  \label{thm:C-ID-sim-of-qubit}
  The simultaneous identification capacity of the noiseless qubit 
  is $C_{\text{ID}}^{\text{sim}}(\id_2) = 1$, whereas the unrestricted 
  identification capacity is $C_{\text{ID}}(\id_2) = 2$.
  
  More generally, for a quantum channel $\cN:A\rightarrow B$, 
  $C_{\text{ID}}^{\text{sim}}(\cN) \leq \log\min\{|A|,|B|\}$, and this is a 
  strong converse bound: for all $\lambda_1,\lambda_2>0$ with 
  $\lambda_1+\lambda_2 < 1$,
  \[
     \limsup_{n\rightarrow\infty} \frac1n \log\log N_{{sim}}(n,\lambda_1,\lambda_2) \leq \log\min\{|A|,|B|\}.
  \]
\end{mytheorem}

\begin{proof}
\input{ProofIdentification}
\end{proof}

\medskip
\begin{mycorollary}
  \label{cor:EB-bound}
  If the channel $\cN:A\rightarrow B$ is entanglement-breaking, or more generally 
  PPT-entanglement-binding, then 
  $C_{\text{ID}}(\cN) \leq \log|A|$, and this is a strong converse 
  bound for all $\lambda_1,\lambda_2>0$ with $\lambda_1+\lambda_2<1$. 
\end{mycorollary}
\begin{proof}
For entanglement-breaking channels, this follows from Theorem \ref{thm:C-ID-sim-of-qubit} 
because the channel is a composition of a destructive measurement $\cM$ with a 
cq-channel $\cN'$, i.e.~$\cN = \cN' \circ \cM$. Thus, 
\[
  C_{\text{ID}}(\cN) \leq C_{\text{ID}}(\cM) = C_{\text{ID}}^{\text{sim}}(\cM) \leq\log|A|,
\]
the middle equality because $\cM$ already outputs the result of a measurement, 
and the same for the maximum code size as a function of error probabilities 
and block length. 

The more general PPT-entanglement-binding case is even simpler: the 
states $\omega^{A_R^nB^n} = (\id_{A_R^n}\ox\cN^{\ox n})\Phi_\sigma^{ \bAn A^n}$ from the proof 
of Theorem \ref{thm:CoveringOneShotSmooth} (cf.~also Theorem \ref{thm:C-ID-sim-of-qubit})
are all PPT, in particular undistillable. 
Thus, the reduction criterion applies, 
$\omega^{A^n_RB^n} \leq \1_{A^n_R} \ox \omega^{B^n}$ \cite{HH:reduction}, hence
\[
  H_{\min}^\epsilon(A_R^{ n}|B^n)_\omega 
    \geq H_{\min}(A_R^{n}|B^n)_\omega
    \geq 0,
\]
and the rest of the argument is as in the proof of Theorem \ref{thm:C-ID-sim-of-qubit}.
\end{proof}

\medskip
\begin{mytheorem}
  \label{thm:C-id-upperbound}
  The classical identification capacity of a general quantum channel 
  $\cN:A\rightarrow B$ is bounded from above as  
  $C_{\text{ID}}(\cN) \leq \log|A| + \widehat{Q}(\cN)$, where 
  $\widehat{Q}(\cN)$ is the strong converse quantum capacity of $\cN$.
\end{mytheorem}
\begin{proof}
This comes from the fact that the resolvability capacity of a quantum 
channel is upper bounded by the strong converse quantum capacity, 
$R(\cN) \leq \widehat{Q}(\cN)$. Then we use the net argument on the 
input as before in the proof of Theorem \ref{thm:C-ID-sim-of-qubit}, but with an auxiliary space of dimension $|R|=2^{n\widehat{Q}(\cN)+o(n)}$. This gives us the desired result.
\end{proof}

%% file: ProofResolvability.tex

We begin with the direct part (upper bound). Consider a $\sigma^A \in \calD(A^{\tensor n})$. From Theorem \ref{thm:CoveringOneShotSmooth}, and
 with 
$\omega^{A_R B}=(\id\ox\calN)\Phi_\sigma^{A_R A}$ in the following optimisation, we have
\[\begin{split}
  \mathcal{R}^0_\epsilon(\cN,\sigma^A) 
    &\leq -H_{\min}^{\epsilon/5}(A_R|B)_\omega + 2\log\frac{20}{\epsilon} \\
    &\overset{(a)}{\leq}  -H_{\max}^{\sqrt{1-\epsilon^4/625}}(A_R|B)_\omega + 2\log\frac{20}{\epsilon} \\
    &\leq -H_{\max}^{1-\epsilon^4/1250}(A_R|B)_\omega + 2\log\frac{20}{\epsilon},
\end{split}\]
where $(a)$ follows from Lemma \ref{lemma:H-max-min}.
(Notice in particular that the right hand side in the first line is always non-negative.)
Now we know from \cite[Prop.~20]{Q-prettystrong} (cf.~Refs.~\citen{BuscemiDatta:one-shot} 
and \citen{DattaHsieh:one-shot}) that the number of qubits transmittable 
via $\cN$ is upper and lower bounded tightly in terms of max-entropies:
\begin{equation}
  \label{eq:one-shot-Q}
  \sup_{\Phi^{A_RA}} -H_{\max}^{\lambda}(A_R|B)_\omega - 4\log\frac{1}{\mu}
    \leq \log N_E(\lambda+\mu|\cN) 
    \leq \sup_{\Phi^{A_RA}} -H_{\max}^{\lambda+\mu}(A_R|B)_\omega,
\end{equation}
with $\omega^{A_RB}=(\id\ox\calN)\Phi^{A_RA}$ as before. Letting here 
$\lambda=1-\epsilon^4/1250$ and $\mu=\epsilon^4/2500$, we thus obtain
\[
  \mathcal{R}^0_\epsilon(\cN,\sigma^A) 
    \leq \log N_E(1-\epsilon^4/2500|\cN) + 2\log\frac{20}{\epsilon} + 4\log\frac{2500}{\epsilon^4}. 
\]
Applying this to $\cN^{\ox n}$ and letting $n\rightarrow\infty$ 
in the regime of small $\epsilon$ shows the upper bound. 


The converse (lower bound) follows from the one-shot soft-covering result obtained in Theorem \ref{thm:CoveringOneShotSmooth}.

%% file: ProofIdentification.tex
The statements about the unrestricted identification capacity are 
from Ref.~\citen{Winter:ID1}, so only the simultaneous identification capacity 
has to be addressed. 
The lower bound $C_{\text{ID}}^{\text{sim}}(\id_2) \geq 1$ follows 
from Ahlswede and Dueck's construction \cite{AhlswedeDueck:ID} showing 
that every transmission capacity can be converted into the same amount 
of double-exponential identification capacity (cf.~Ref.~\citen{Loeber:PhD}). 

As a channel $\cN:A\rightarrow B$ can be simulated by $\id_A$ followed 
by post-processing at the receiver, and also by pre-processing at the 
sender followed by $\id_B$, every simultaneous identification code for 
$\cN$ gives rise to a simultaneous identification code for an ideal 
channel. Thus, we have to prove the upper bound only for $\id_A$. 
For that, consider a simultaneous identification code for $\id_A^{\ox n} = \id_{A^n}$
with $N$ messages and error probabilities $\lambda_1,\lambda_2>0$ such 
that $\lambda_1+\lambda_2 < 1$. To message $m$ is associated the state 
$\rho_m$ on $A^n$, and the binary decision POVM $(D_m,\1-D_m)$ on 
$A^n$.
The simultaneity condition requires that all these POVMs are 
compatible, i.e.~there is a single measurement $(M_y)_{y\in\cY}$ with 
an arbitrary number of outcomes such that each $(D_m,\1-D_m)$ is obtained 
as a coarse-graining of it. Define the qc-channel $\cM:A^n\rightarrow Y$ 
corresponding to the measurement, $\cM(\rho) = \sum_y (\Tr\rho M_y)\proj{y}$. 
Since the given code is ipso facto an identification code for $\cM$, 
we are motivated to apply Theorem \ref{thm:CoveringOneShotSmooth} to 
this channel, with $\delta=\eta=\frac{1}{24}(1-\lambda_1-\lambda_2)$. 
This gives us for each $m$ a state $\sigma_m$ of rank $r$ on $A^n$ with 
\[
  \frac{1}{2}\left\| \cM(\sigma_m)-\cM(\rho_m) \right\|_1 \leq \frac13(1-\lambda_1-\lambda_2) =: \Delta, 
\]
and $\log r \leq \left[-H_{\min}^\epsilon(A^{n}_R|Y)_\omega - 2\log\eta\right]^+$, where
\[
  \omega^{A^{n}_RY} = \sum_y \Tr_{A^n}\round{ (\1\ox M_y)\Phi_m^{A_R^n A^n}}\ox \proj{y}^Y,
\]
and $\Phi_m^{A_R^n A^n}$ is a purification of $\rho_m$.
The crucial insight now is that 
\[
  H_{\min}^\epsilon(A^{n}_R|Y)_\omega 
    \geq H_{\min}(A^{n}_R|Y)_\omega
    \geq 0, 
\]
since 
$\omega^{A^{n}_RY}$ 
is a cq-state, in particular separable. 
Thus, $r \leq \frac{2}{\eta^2}$, and we can choose purifications 
$\ket{\zeta_m} \in A^n\ox R$ of $\sigma_m$ with $|R|=\frac{2}{\eta^2}$.
At the same time, the collection of states $\sigma_m$, with the same 
measurements $(D_m,\1-D_m)$, is an identification code with error probabilities 
$\lambda_i' = \lambda_i + \Delta$. 

Next, we choose an $\delta$-net $\cZ$ of pure states on $A^n\ox R$, which is 
known to exist with cardinality $|\cZ| \leq \left(\frac{5}{\delta}\right)^{2|A^n||R|}$ \cite[Lemma III.6]{hayden2006aspects}. 
For each $m$, we can find an element $\xi_m\in\cZ$ with 
$\frac12\|\zeta_m-\xi_m\|_1 \leq \delta$, and hence the collection of states 
$\tau_m^{A^n} = \Tr_R \xi_m$ is an identification code with error probabilities 
$\lambda_i'' = \lambda_i + \Delta + \delta$. 
As $\lambda_1''+\lambda_2'' = \lambda_1+\lambda_2 + 2\Delta + 2\delta 
= 1-\frac14(1-\lambda_1-\lambda_2) < 1$, 
the states $\xi_m$ must be pairwise distinct, and that gives us 
\[
  N \leq |\cZ| \leq \left(\frac{5}{\delta}\right)^{2|A^n||R|},
\]
i.e.
\[
  \log\log N \leq n\log|A| + \log\frac{2304}{(1-\lambda_1-\lambda_2)^2} + \log\log\frac{120}{1-\lambda_1-\lambda_2},
\]
which is the desired bound. 

%% file: petz_reference.tex
Consider a state $\rho^B \in \calD(B)$ and a channel $\calN_V: B \rightarrow A$. Let $V:B \rightarrow A\tensor E$ be an isometric extension of $\calN_V$ with $\dim(E) > \dim(B)$. Let the spectral decomposition of $\rho^A \deq \calN_V(\rho^B)$ and $\rho^B$ be given by  $\rho^A = \sum_{a\in \calA}\lambda_a \ketbra{a}$, and $\rho^B = \sum_{b\in \calB}\mu_b \ketbra{b}$, for some finite sets $\calA$ and $\calB$.
From Eq.~\eqref{eq:WV_postrefmap}, we have
\begin{align*}
    W\ket{a}^{\bA} = \sum_{b} \sqrt{\frac{\mu_b}{\lambda_a}} \ket{b}^{\bB} \tensor \langle a|V|b\rangle,
\end{align*}
where $\ket{a}^{\bA} \deq (\1_{\bA}\tensor \bra{a}^A)\ket{\Gamma}_{\bA A}$ and $\ket{b}^{\bB} \deq (\1_{\bB}\tensor \bra{b}^B)\ket{\Gamma}_{\bB B}$ for fixed maximally entangled states $\ket{\Gamma}_{\bA A}$ and $\ket{\Gamma}_{\bB B}$ on the Hilbert spaces $\bA\tensor A$ and $\bB\tensor B$, respectively. 

Consider an arbitrary density operator $\sigma^{\bA} \in \calD({\bA}),$ and let $\sigma^{\bA} = \sum_{a,a'} \theta_{aa'}\ket{a}\bra{a'}^{\bA}$ be its matrix representation, and then  $\sigma^{A} = (\sigma^{\bA})^T =\sum_{a,a'} \theta_{aa'}\ket{a'}\bra{a}^{A}.$  Then,
\begin{align}
    \calN_W(\sigma^{\bA}) & = \sum_{b,b'}\sum_{a,a'} \sqrt{\frac{\mu_b\mu_{b'}}{\lambda_a\lambda_{a'}}} \theta_{aa'} \ket{b}\bra{b'}^{\bB}  \sq{_B\langle b'|V^\dagger(|a'\rangle\langle a| \tensor \1_E)V|b\rangle_B} \nonumber \\
    & = \sum_{b,b'}\sum_{a,a'} \sqrt{{\mu_b\mu_{b'}}} \ket{b}\bra{b'}^{\bB} \sq{_B\langle b'|V^\dagger\round{(\rho^A)^{-1/2} \sigma^A(\rho^A)^{-1/2} \tensor \1_E}V|b\rangle_B} \nonumber \\
    & = \sum_{b,b'}\sum_{a,a'} \sqrt{{\mu_b\mu_{b'}}} \ket{b}\bra{b'}^{\bB}  \sq{_B\langle b'|\calN_V^\dagger\round{(\rho^A)^{-1/2} \sigma^A(\rho^A)^{-1/2}}|b\rangle_B} \nonumber \\
    &\stackrel{a}{=} \sum_{b,b'}\sum_{a,a'} \sqrt{{\mu_b\mu_{b'}}}  \ket{b}\bra{b'}^{\bB}  \sq{_{\bB}\langle b|\round{\calN_V^\dagger\round{(\rho^A)^{-1/2} \sigma^A(\rho^A)^{-1/2}}}^T|b'\rangle_{\bB}} \nonumber \\
    &= (\rho^{\bB})^{1/2} \round{\calN_V^\dagger\round{(\rho^A)^{-1/2} \sigma^A(\rho^A)^{-1/2}}}^T(\rho^{\bB})^{1/2}, \nonumber
\end{align}
where $(a)$ follows from the definition of transpose.
Therefore the action of a posterior reference map is equal to suitably defined Petz recovery channel acting on the reference Hilbert space.

%% file: lemma_proofs.tex
\subsection{Proof of Lemma \ref{lem:coveringAndcompression}}\label{proof:coveringAndcompression}
The first part of the lemma follows using the monotonicity of trace norm by partial tracing the states in Eq.~\eqref{def:protocolError} over the subsystem ${A}^{\tensor n},$ observing that $ \Tr_{A^n}{\omega^{\bBn A^n}} = \round{\rho^{B_R}}^{\tensor n}$, and noting that the restriction of decoding being an isometry maintains the rank of $\omega^{A_R^n}$ to be at most $\Theta$.

For the converse direction, we perform the following analysis. Given a pair $(\rho^{\bB}, \calN_W)$, and an $(n,\Theta,\epsilon)$ covering code $\sigma^{A_R^n}$, let 
$\Psi^{A_R^n A^n}_{\sigma}$ be the canonical purification of 
$\sigma^{A_R^n}$. By Definition \ref{def:qu covering code} of a covering code, 
\begin{align}
  \frac12 \left\|\round{\rho^{B_R}}^{\tensor n} - \calN_W^{\tensor n}(\sigma^{\bAn})\right\|_1 \leq \epsilon.
  \label{eq:qsc_error}
\end{align}
Note that $|\sigma^{A_R^n}|=\Theta$. 
Define $\calH_{\sigma} \subset {A^n}$
as the support of $\Tr_{A_R^n} [\Psi^{A_R^n A^n}_{\sigma}]$, while noting that $\dim(\calH_{\sigma})=\Theta$. 
Let $W: {A_R} \rightarrow {B_R} \otimes E$ be a Stinespring extension of $\calN_W$ with $\mbox{dim}(E) \geq \mbox{dim}({B_R})$. 
Consider the quantum state 
\[
\omega^{B_R^n E^n A^n} \deq (W^{\tensor n} \otimes \1_{A^n}) \Psi^{A_R^n A^n}_{\sigma} (W^{\tensor n} \otimes \1_{A^n})^{\dagger},
\]
defined on ${B_R^n} \otimes {E^n} \otimes \calH_{\sigma}$. 
Let $\tilde{V}: {B^n} \rightarrow \calH_{\sigma} \otimes {E^n}$ be the posterior reference isometry of 
$W^{\tensor n}$ with respect to $\omega^{B_R^n}$. 
Note that $\tilde{V}$ is a non-product $n$-letter isometry. Using the property of the reference map, 
we have 
\[
\omega^{B_R^n E^n A^n}= (W^{\tensor n} \otimes \1_{A^n}) \Psi^{A_R^n A^n}_{\sigma} (W^{\tensor n} \otimes \1_{A^n})^\dagger
= (\1_{B_R^n} \otimes \tilde{V}) \Psi^{B_R^n B^n}_\omega (\1_{B_R^n} \otimes \tilde{V})^\dagger,
\]
where $\Psi^{B_R^n B^n}_\omega$ is the canonical purification of $\omega^{\bBn} = \calN_W^{\tensor n} (\sigma^{A_R^n})$. 
We use the CPTP map $\calN_{\tilde{V}}^{(n)}$ given by $\Tr_{E^n}[\tilde{V} \cdot \tilde{V}^{\dagger}]$ as the encoder-decoder pair
${\calD}^{(n)} \circ {\calE}^{(n)}$
for the lossy source coding protocol resulting in the quantum state 
\[
\tau^{B_R^n A^n} \deq 
(\id_{\bBn}\tensor \calN_{\tilde{V}}^{(n)})(\Psi_\rho^{B_R B})^{\tensor n} , 
\]
where $\Psi^{B_R B}_{\rho}$ is the canonical purification of $\rho^{B_R}$. 
Furthermore, define the quantum state 
\[
\lambda^{B_R^n A^n}
\deq  (\calN_W^{\tensor n}  \otimes \id_{A^n}) \Psi^{A_R^n A^n}_\tau, 
\]
where $\Psi^{A_R^n A^n}_\tau$ is the canonical purification of $\tau^{A^n}=\Tr_{B_R^n}[\tau^{B_R^n A^n}]$. 
Consider the trace distance between the states $\tau$ and $\lambda$ as follows:
\begin{align*}
\left\| \tau^{B_R^n A^n} - \lambda^{B_R^n A^n}\right\|_1 &\leq 
\left\|  \omega^{B_R^n A^n} - \tau^{B_R^n A^n} \right\|_1 +  
\left\| \omega^{B_R^n A^n} -  \lambda^{B_R^n A^n} \right\|_1 \\
& = \left\| (\id_{\bBn}\tensor \calN_{\tilde{V}}^{(n)})\round{\Psi_\omega^{B_R^n B^n} -(\Psi^{B_R B}_\rho)^{\tensor n}} \right\| \\
&\phantom{====}
 + \left\| (\calN_W^{\tensor n}\tensor \id_{A^n})\round{\Psi_\sigma^{A_R^n A^n} -  \Psi_\tau^{A_R^n A^n}} \right\|_1  \\
&\overset{(a)}{\leq} \| \Psi_\omega^{B_R^n B^n} -(\Psi_\rho^{B_R B})^{\tensor n} \|_1
+\|\Psi_\sigma^{A_R^n A^n} -  \Psi_\tau^{A_R^n A^n}\|_1 \\
&\overset{(b)}{\leq} 2\sqrt{\| \omega^{\bBn} -\rho^{B_R^{\tensor n}} \|_1}+2 \sqrt{\| \sigma^{A^n} -  \tau^{A^n}\|_1} \\
&\overset{(c)}{\leq} 2\sqrt{\epsilon} + 2 \sqrt{\| \calN_{\tilde{V}}^{(n)} (\omega^{B^n})-
\calN_{\tilde{V}}^{(n)} (\rho^{B^{\tensor n}})
\|_1}\\
&\overset{(d)}{\leq} 2\sqrt{\epsilon} + 2 \sqrt{\|  \omega^{B^n}-
\rho^{B^{\tensor n}} \|_1} \\
&\overset{(e)}{=} 2\sqrt{\epsilon} + 2 \sqrt{\|  \calN_W^{\tensor n}(\sigma^{\bAn})-
\rho^{B_R^{\tensor n}} \|_1} 
 \overset{(f)}{\leq} 4\sqrt{\epsilon},
\end{align*}
where (a) follows from the monotonicity of trace distance,  (b) follows from the Lemmas \ref{lem:relationshipTraceFidelity} and \ref{lem:closenessofPurification}, and appropriate identification of canonical purifications, (c) uses the covering bound Eq.~\eqref{eq:qsc_error} and defines $\omega^{B^n}\deq \Tr_{\bBn}{\Psi^{\bBn B^n}_\omega}$, and (d) again uses monotonicity of trace distance, (e) follows from the invariance of trace norm under transposition, and (f) again uses the covering property, Eq.~\eqref{eq:qsc_error}. Since $\calN_{\tilde{V}}^{(n)}$ is a CPTP map from ${\bBn}$ to $\calH_\sigma$, and $|\calH_\sigma| = \Theta$, the lossy source code also satisfies the rate constraint. This completes the proof.

%% file: lossyConverseOneShot.tex

Given a pair $(\rho^B, \calN_W)$, and any $\epsilon \in (0,1)$,  consider an $(1,\Theta,\epsilon)$ lossy compression protocol with an encoding CPTP map $\calE$ and a decoding CPTP map $\calD$ that satisfies 
\begin{align}
\label{eq:cons_conv1}
\frac12 \left\|\omega^{\bB A} - \upsilon^{\bB A} \right\|_1 \leq \epsilon,
\end{align}
where $\omega^{\bB A} \deq (\id \tensor \calD\circ \calE) (\Psi_{\rho}^{\bB B})$, 
\[
\upsilon^{\bB A} =\Tr_{E} \upsilon^{{B}_R A E} \deq \Tr_{E}(\bV \otimes \1) \Psi_{\omega}^{\bA A}(\bV \otimes \1)^{\dagger},
\] and $\Psi_{\rho}^{\bB B}$ and $\Psi_{\omega}^{A \bA}$ are the canonical purifications of $\rho^B$ and $\omega^{A}$, respectively, and $\bV$ is the Stinespring's dilation of the CPTP map $\calN_W$. 
Let $\omega^{A_R}\deq \Tr_{A} \Psi_\omega^{A A_R}$.

Let $M$ denote the quantum state at the output of the encoder.
Let $V_\calE: {B} \rightarrow 
M \tensor {\tilde{E}_1}$ and 
$V_\calD:{M} \rightarrow {A} \tensor {\tilde{E}_2}$ be Stinespring dilations of encoding and decoding maps ${\calE}$ and 
${\calD}$, respectively, such that $\dim {\Tilde{E}_1} \geq \dim M$ and $\dim {\Tilde{E}_2} \geq \dim A$,  as shown in Figure \ref{fig:q_converse}(a).
Let
\begin{align}
    {\omega_1}^{\bB M \tilde{E}_1} &\deq (\1_{\bB}\tensor V_\calE)(\Psi_{\rho}^{\bB B})(\1_{\bB}\tensor V_\calE)^\dagger,\nonumber
     \\
     \quad {\omega}^{\bB \tilde{E}_1 \tilde{E}_2 A} &\deq (\1_{\bB\Tilde{E}_1}\tensor V_\calD)({\omega_1}^{\bB \tilde{E}_1 M})(\1_{\bB\Tilde{E}_1}\tensor V_\calD)^\dagger. \nonumber
\end{align}
\begin{figure}[ht]
    \centering
    \includegraphics[width=\textwidth]{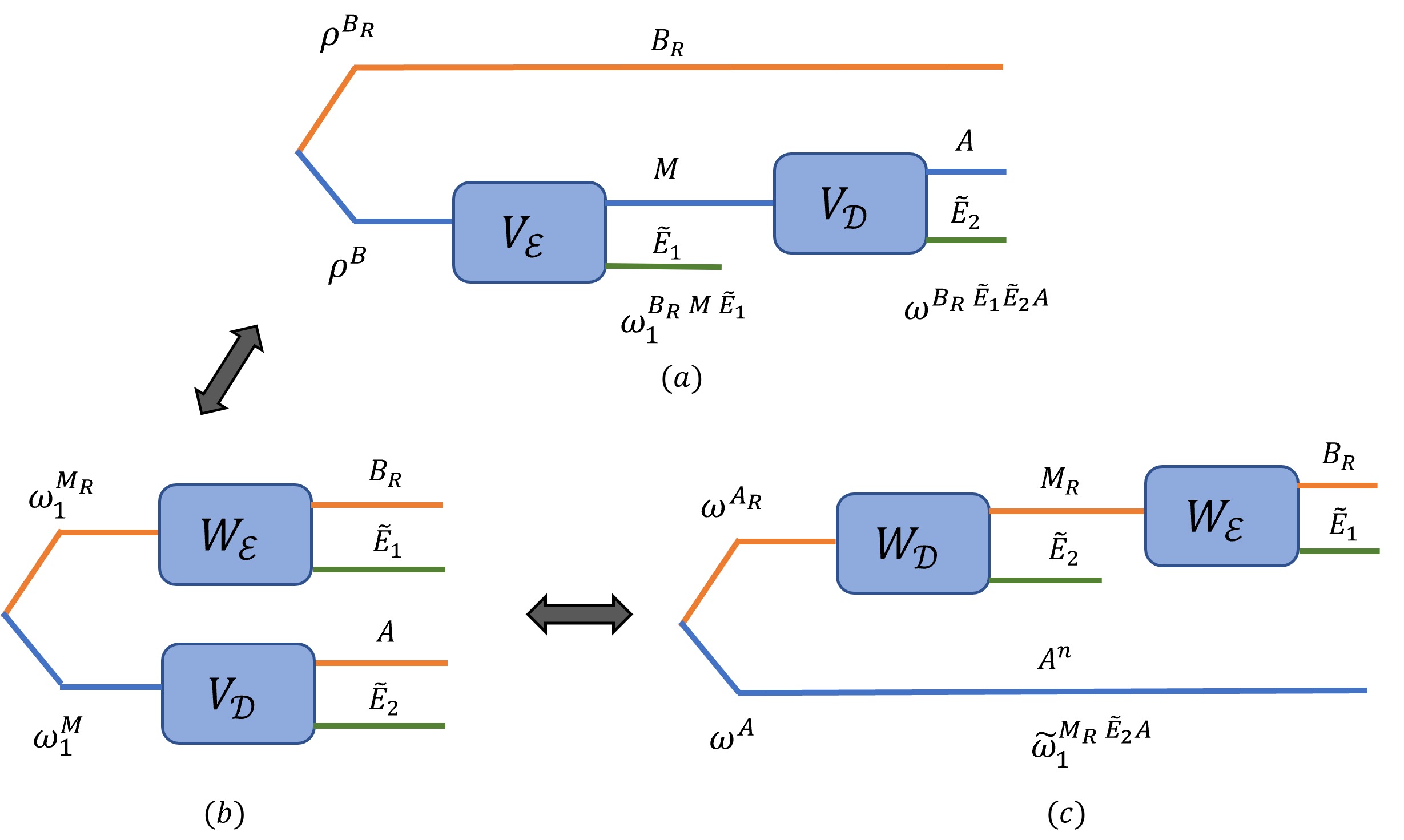}
    \vspace{-0.25in}
    \caption{A One-shot lossy quantum source coding protocol and the associated CPTP maps and their Stinespring dilations.}
    \label{fig:q_converse}
    \vspace{-0.2in}
\end{figure}

\noindent
Now let $\Psi_{\omega_1}^{M_R M}$ denote the canonical purification of the quantum state $\omega_1^M$, and let
$W_\calE^{(n)}: {M_R} \rightarrow {\bB}\tensor {\Tilde{E}_1}$ denote the posterior reference isometry (see Definition \ref{def:VBar}) of
$V_\calE$ with respect to $\omega_1^M$, as shown in Figure \ref{fig:q_converse}(b). Moreover, let 
$W_\calD: {\bA} \rightarrow {M_R} \tensor {\Tilde{E}_2}$ denote the posterior reference isometry of
$V_\calD$ with respect to $w^{A}$, as shown in Figure \ref{fig:q_converse}(c). Let $\calN_{W_{\calE}}(\cdot)=\Tr_{\tilde{E_1}}(W_{\calE} \cdot (W_{\calE})^{\dagger})$
and $\calN_{W_{\calD}}(\cdot)=\Tr_{\tilde{E_2}}(W_{\calD} \cdot (W_{\calD})^{\dagger})$
be the induced CPTP maps. 
Let \[
\tilde{\omega}_1^{M_R \tilde{E}_2 A} \deq (W_\calD\tensor \1_{A})(\Psi_{\omega}^{\bAn A})(W_\calD\tensor \1_{A})^\dagger.
\]
By the definition of the posterior maps, we have
\begin{align}
    {\omega}^{B_R \tilde{E}_1\tilde{E}_2A} = (W_\calE \tensor \1_{\Tilde{E}_2 A}) \tilde{\omega}_1^{M_R \tilde{E}_2 A}  (W_\calE \tensor \1_{\Tilde{E}_2 A}) ^\dagger. \nonumber
\end{align}
We have the following chain of inequalities:
\begin{align}
    \log\Theta & {\geq} \log \min \{\mbox{rank}\round{\tilde\omega_1^A},\mbox{rank}\round{\tilde\omega_1^{M_R}}\}\nonumber \\
    & \overset{(a)}{\geq} -\hmin(A|M_R)_{\tilde\omega_1} \nonumber \\
    & \overset{(b)}{\geq}  -\hmin(A|B_R)_{\omega}
    \nonumber \\
    & \overset{(c)}{\geq} -\hmin^{\sqrt{\epsilon}}(A|B_R)_{\upsilon} 
       \nonumber \\
    & \overset{(d)}{=} -\hmin^{\sqrt{\epsilon}}(A|B_R)_{\upsilon(\omega^{A_R})}
       \nonumber \\
    & \overset{(e)}{\geq} -\sup_{\sigma^{A_R}:\|\rho^{B_R} - \calN_W(\sigma^{A_R})\|_1 \leq \epsilon}\hmin^{\sqrt{\epsilon}}(A|B_R)_{\upsilon(\sigma^{A_R})}, \nonumber
\end{align}
where (a) follows from \cite[Lemma 5.11]{tomamichel2015quantum},
(b) follows from the quantum data processing inequality for min-entropy \cite[Corollary 3.5]{tomamichel2015quantum},
(c) follows from condition \eqref{eq:cons_conv1}, which implies that the corresponding purified distance is no greater than $\sqrt{\epsilon}$,
(d) follows by defining 
$ \upsilon(\omega^{A_R}) \deq (\calN_W \tensor \1_A)\Psi_{\omega}^{A_R A}$, and (e) again follows from condition (\ref{eq:cons_conv1}), the monotonicity of the trace distance, and by noting that  $\frac12 \left\| \Tr_A\round{\upsilon(\omega^{A_R})} -\rho^{\bB} \right\|_1 \leq \epsilon$.

\subsection{Proof of the converse in Theorem \ref{thm:mainResult_lossy}}
\label{sec:converse_asymp_lossy}
For the converse, we proceed as follows. Let $R$ be achievable. Using the result from Theorem \ref{thm:CoveringOneShotLossy}, 
for any $\epsilon>0$, we get
\begin{align*}
    nR &\geq \log \Theta -n \epsilon \geq \inf_{\sigma^{A^n_R}:\|(\rho^{B_R})^{\tensor n} - \calN_W^{\tensor n}(\sigma^{A^n_R})\|_1 \leq \epsilon}-\hmin^{\sqrt{\epsilon}}(A^n|B^n_R)_{\rho^{B_R^n A^n} }-n \epsilon, 
\end{align*}
where $\rho^{B_R^n A^n} \deq ( \id \tensor\calN_W^{\tensor n} )\Phi_{\sigma}^{A^n_R A^n}$ with $\Phi_{\sigma}^{A^n_R A^n}$ denoting a purification of $\sigma^{A^n}$.

Using \cite[Lemma 6.12]{tomamichel2015quantum}, we note that there exists an embedding from $A^n$ to $A'$ and a normalized state $\tau^{A'B^n}$ in $\mathscr{B}^{\sqrt{\epsilon}}(A'B_R^n; \rho^{B_R^nA'}))$ such that $\hmin^{\sqrt{\epsilon}}(A^n_R|B^n)_{\rho^{B_R^n,A'}} = \hmin(A'_R|B^n)_{\tau}$, where $\rho^{B_R^nA'}$ is the corresponding embedding of $\rho^{B_R^nA^n}$, with $\dim(A') = 2(\dim A)^n$.

We next use the fact that $\hmin(\cdot|\cdot) \leq S(\cdot|\cdot)$, and  Lemma \ref{lem:relationshipTraceFidelity} and the Alicki-Fannes-Winter continuity inequality for the conditional entropy \cite[Thm.~11.10.3]{wilde_arxivBook}
to infer that 
\[
 -\hmin^{\sqrt{\epsilon}}(A^n|B^n_R)_{\rho^{B_R^n A^n} } = \hmin(A'_R|B^n)_{\tau} \geq - S(A'|B^n_R)_{\tau} \geq -S(A^n|B^n_R)_{\rho} -n \tilde{\epsilon}_1, 
\]
where $\tilde{\epsilon}_1 \deq 2\sqrt{\epsilon}\log(2|A|) + \frac{1+\sqrt{\epsilon}}{n}h_2\left(\frac{\sqrt{\epsilon}}{1+\sqrt{\epsilon}}\right)$, and we use the invariance of conditional entropy under embeddings.
Thus we have 
\[
nR \geq \inf_{\sigma^{A^n_R}:\|(\rho^{B_R})^{\tensor n} - \calN_W^{\tensor n}(\sigma^{A^n_R})\|_1 \leq \epsilon}-S(A^n|B^n_R)_{\rho} -n \tilde{\epsilon}_1 -n \epsilon.
\]
Observe that the above equation is now similar to \eqref{eq:covering_converse} obtained in the converse of Theorem \ref{thm:asympCovering}.
Using similar arguments as in the converse of Theorem \ref{thm:asympCovering}
we obtain the desired result.